\documentclass[10pt]{article} 
\usepackage[margin=3cm,bottom=3cm]{geometry} 
\pdfoutput=1

\usepackage[utf8]{inputenc}
\usepackage[T1]{fontenc}
\usepackage{amsmath,mathtools,amssymb,amsfonts,commath} 
\usepackage{amsthm}
\usepackage{etoolbox}
\usepackage{bm}
\usepackage[table]{xcolor}
\usepackage{graphicx}
\usepackage{dcolumn}
\usepackage{url}
\usepackage{pifont}
\usepackage[ruled, lined, linesnumbered, commentsnumbered, longend]{algorithm2e}

\usepackage{tikz}
\usepackage{braket}
\usepackage[colorlinks=true, linkcolor=blue, citecolor=blue, urlcolor=blue]{hyperref}

\usepackage[capitalise]{cleveref}

\usepackage[multiple]{footmisc} 

\usepackage{blindtext}

\usepackage{tikz}
\usetikzlibrary{calc}
\usetikzlibrary{positioning}
\usepackage{quantikz}
\usepackage{enumitem}
\setlist[enumerate]{nosep}
\setlist[itemize]{nosep}
\usepackage{multirow}
\usepackage{booktabs}
\usepackage{pgfplots}
\pgfplotsset{compat=1.18}
\usepackage{subcaption}

\usepackage[style=numeric, maxnames=3, minnames=1, sorting=none, citestyle=numeric-comp,]{biblatex}
\newcommand\identity{1\kern-0.25em\text{l}}

\newtheorem{theorem}{Theorem}
\newtheorem{lemma}[theorem]{Lemma}
\newtheorem{corollary}[theorem]{Corollary}

\theoremstyle{plain}
\newtheorem{definition}[theorem]{Definition}

\newtheorem{problem}[theorem]{Problem}

\DeclareMathOperator{\RR}{\mathbb{R}}
\DeclareMathOperator{\EX}{\mathbb{E}}
\DeclareMathOperator{\domain}{\mathcal{D}} 
\newcommand{\imi}{\mathrm{i}} 

\newcommand\restr[2]{{
  \left.\kern-\nulldelimiterspace 
  #1 
  \vphantom{\big|} 
  \right|_{#2} 
  }}

\newcommand{\PF}{\text{PF}}

\newcommand{\RTE}{\text{RTE}}
\newcommand{\rpf}{r_{\mathrm{PF}}}
\newcommand{\epf}{\epsilon_{\mathrm{PF}}}
\newcommand{\Bpf}{B_{\mathrm{PF}}}
\newcommand{\Brte}{B_{\mathrm{RTE}}}
\newcommand{\nmax}{n_{\mathrm{max}}}
\newcommand{\Nrte}{N_{\mathrm{RTE}}}
\usepackage{authblk}

\title{The Practicality of Randomized Quantum Linear Systems Solvers}
\author[1,2,3,4,*]{Siddharth Hariprakash}
\author[5]{Roel Van Beeumen}
\author[1]{Katherine Klymko}
\author[1,*]{Daan Camps}

\affil[1]{National Energy Research Scientific Computing Center, Lawrence Berkeley National Laboratory, Berkeley, CA 94720, USA}
\affil[2]{Leinweber Institute for Theoretical Physics and Department of Physics, University of California, Berkeley, California 94720, USA}
\affil[3]{Physics Division, Lawrence Berkeley National Laboratory, Berkeley, CA 94720, USA}
\affil[4]{BlueQubit Inc, San Francisco, CA 94105, USA}
\affil[5]{Applied Mathematics and Computational Research Division, Lawrence Berkeley National Laboratory, Berkeley, CA 94720, USA}
\affil[*]{Corresponding author: siddharth\_hari@berkeley.edu}
\affil[*]{Corresponding author: daancamps@gmail.com}

\date{}
\begin{document}

\maketitle

\begin{abstract}

Randomized quantum algorithms have been proposed for quantum linear algebra with the goal of constructing shallower circuits than methods based on block encodings, and have been speculated to offer benefits in the early fault-tolerant era. In this work, we derive explicit, non-asymptotic error bounds on every algorithmic parameter of a randomized quantum linear systems solver that combines sampling from a Fourier series with Hamiltonian simulation, and confirm these bounds numerically. We show that even a $4 \times 4$ instance with condition number $\kappa = 100$ requires on the order of $10^{15}$ non-Clifford gates to converge, calling into question the practicality of such randomized schemes. Comparing the two Hamiltonian-simulation kernels, product formulas (PFs) and the random Taylor expansion (RTE), both our analytical bounds and experiments show RTE reaches a given target error with roughly an order of magnitude smaller total gate budget, although neither is practical. Our analysis bridges asymptotic proposals and hardware implementation.

\end{abstract}

\section{Introduction}
\label{sec:intro}

Quantum computers hold significant promise for solving complex computational problems in scientific computing, particularly within the physical sciences~\cite{qc_chem, qc_hep, Bauer:2023qgm, Bauer:2022hpo, qc_materials}. This has inspired a broader search for quantum algorithms that can effectively solve \emph{classical} computational problems arising in combinatorial optimization, machine learning, and linear algebra~\cite{Montanaro2016, Morales:2024azh}. In this work, we focus on the latter and study the topic of \emph{quantum linear algebra}.

Most existing approaches to quantum linear algebra rely on a fully coherent, oracle-based data access model with Hamiltonian simulation~\cite{hhl,Kirby:2025iiw} and block-encoding access~\cite{qsvt} as two of the most common assumptions. Hamiltonian simulation is a key quantum algorithmic primitive; it is known to be in the BQP complexity class, and is efficiently implementable on quantum computers for $k$-local and sparse Hamiltonians~\cite{trot_theory, qdrift, Berry2006, Berry2015, Low2017, Low2019hamiltonian}. Consequently, Hamiltonian simulation has become a widespread quantum routine and a foundation for many subsequent quantum algorithms. As a subroutine, it provides query access to the time evolution operator $e^{-\imi H t}$ with controllable error $\epsilon$. In particular, the first quantum algorithm to solve linear systems of equations proposed by Harrow, Hassidim, and Lloyd (HHL) makes use of Hamiltonian simulation~\cite{hhl}. On the other hand, a block encoding is an embedding of a non-unitary operator into a larger unitary, and is central to algorithms based on the quantum singular value transformation (QSVT)~\cite{qsvt, unification}. Block encodings provide query access to the operator itself or to polynomial transformations of the operator through QSVT and its variants. Quantum circuits for block encodings can be constructed explicitly for certain classes of highly structured, data-sparse problems~\cite{doi:10.1137/22M1484298, Sunderhauf2024blockencoding,Kane:2024odt,Hariprakash:2023tla,Rhodes:2024zbr,LIU2025102480} without relying on quantum random access memory (QRAM)~\cite{clader22}, or from linear combinations of unitaries (LCU)~\cite{lcu,dellachiara2025}, which is efficient for operators that are sparse in the Pauli basis.

Coherent data access via block encoding models generally requires deep and complicated quantum circuits that will remain challenging to implement for near-term quantum computers, even if some degree of fault-tolerance is achieved. In an effort to bypass the challenge of constructing block encodings, a new class of randomized algorithms have been proposed in the context of quantum linear algebra. For this approach, the computational problem is reformulated such that a classical sampling procedure combined with a Hamiltonian simulation subroutine can be used to provide Monte Carlo estimates to the desired solution upon averaging over all estimates. It has been shown that a key advantage of this approach is that the circuit complexity of each individual sample can be reduced at the cost of increasing the total number of samples required to reach a certain precision. Prominent examples of randomized quantum algorithms include methods for Hamiltonian simulation~\cite{qdrift,Granet2024}, for statistical phase estimation~\cite{Wan:2021non}, and for computing more general matrix functions~\cite{Wang:2023els,Chakraborty:2023vtr,Wang:2025atx,Martyn:2024hwl}.

\begin{figure}[t]
    \centering
    \includegraphics[width=0.7\textwidth]{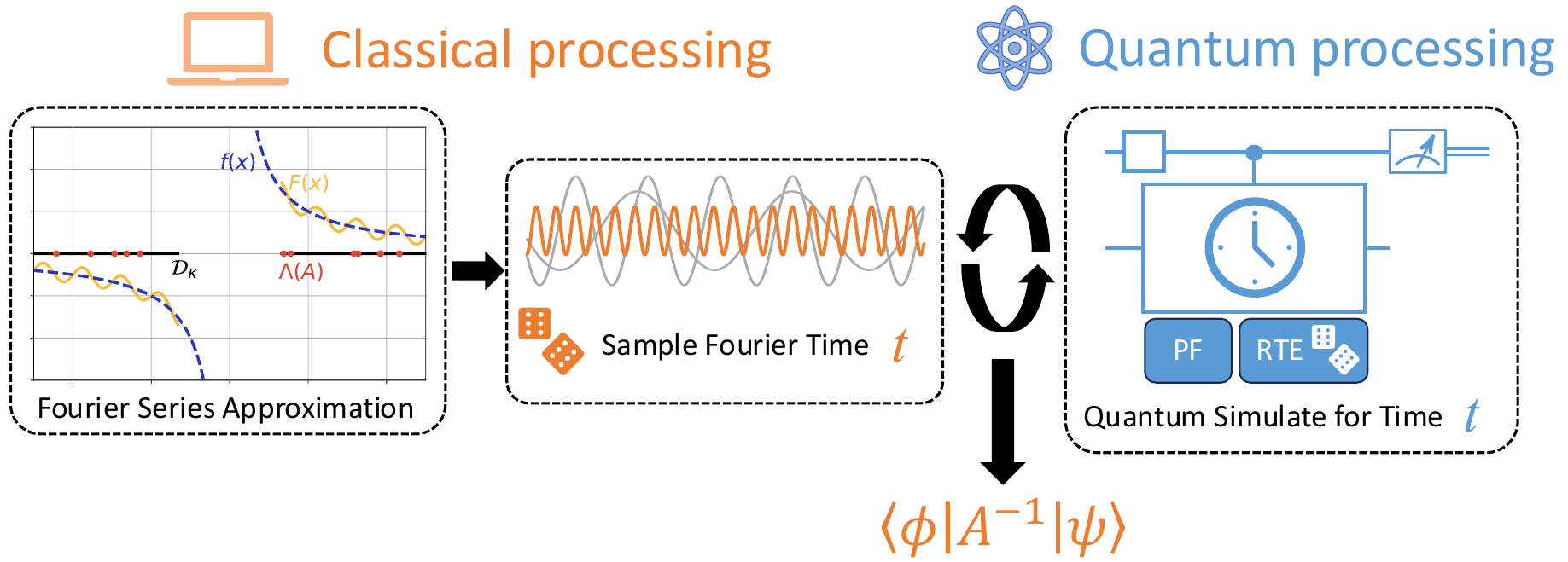}
    \caption{Randomized framework for sampling from solutions to the quantum linear systems~(QLS) problem $\bra{\phi} A^{-1} \ket{\psi}$ for a Hermitian matrix $A$, consisting of a classical processing step (\textbf{left}) which constructs a provably accurate Fourier series approximation to the inverse function, followed by repeatedly sampling a Fourier time $t$ and constructing an approximate time evolution operator $U \approx e^{-\imi At}$ using either a Product Formula (\PF) or Random Taylor Expansion (\RTE). Next the quantum processing step (\textbf{right}) uses $\log(N) + 1$ qubits to measure $\bra{\phi} U \ket{\psi}$ for each sample. The statistical average over all samples converges to $\bra{\phi} A^{-1} \ket{\psi}$.}
    \label{fig:overview_fig}
\end{figure}

In this work, we build on~\cite{Wan:2021non,Wang:2023els} and present an explicit randomized quantum algorithm for quantum linear algebra problems based on Fourier series approximations. We consider the task of estimating scalar properties of matrix functions $f(A)$ for an $N \times N$ Hermitian operator $A$, i.e., quantities of the form $\bra{\phi}f(A)\ket{\psi}$ (We do not consider the problem of computing the required normalization factor in this work, see \cite{Wang:2023els} for a careful treatment of this topic). We primarily focus on the quantum linear systems~(QLS) problem, that is $f(A) = A^{-1}$, but provide a brief overview of how to extend the methods that we present to arbitrary functions. The proposed method is summarized in \cref{fig:overview_fig} and preserves the key advantages of~\cite{Wan:2021non,Wang:2023els}: only a single ancilla qubit leading to a $\log(N) + 1$ qubit complexity for $N$-dimensional problems and tunable circuit and sampling complexities. 

We note that all previous work in the space of Fourier series based randomized QLS
solvers provide only asymptotic complexities. While useful, it is unclear what the
actual end-to-end resource requirements are (throughout, we use ``end-to-end'' to
mean explicit in all resources associated with the matrix $A$: the Fourier-series
construction, the sampling procedure, and the controlled time-evolution
circuits, given state-preparation oracles for $\ket{\psi}$ and $\ket{\phi}$, whose
implementation cost we treat as a separate, additive contribution to the per-sample gate complexity). The problem of constructing the Fourier series is often left as an open question in previous work and hence its contribution to the final algorithmic complexities have not been thoroughly studied. In this work, we accurately characterize the error made by approximating the inverse function by one possible choice for the Fourier series approximation proposed in recent work \cite{qrt24}. We also place upper bounds on the variance introduced as a result of sampling individual terms from this Fourier series. Thus, this allows us to compute end-to-end resource requirements for the full randomized algorithm including classical and quantum subroutines. We show that even for moderately sized problems, our analysis puts into question the practicality of such randomized schemes.  

After defining the problem and giving an overview of the randomized procedure, our main contributions include:
\begin{enumerate}[label=\roman*]
    \item A detailed and complete analysis of the truncation and discretization errors (\cref{thm:truncation_parameters,thm:fourier_parameters}) stemming from approximating the inverse function as a Fourier series. This includes \emph{non-asymptotic} choices for the parameters used to construct
the Fourier series given some error tolerance, using the discretization scheme
of~\cite{qrt24}, which requires exponentially fewer quadrature nodes---logarithmic rather than polynomial in $1/\epsilon$---than a
low-order discretization.
    \item \emph{Non-asymptotic} algorithmic resource estimates given an error budget. Specifically, we compute the number of samples required by the randomization procedure as well as the non-Clifford gate depth of each sample (\cref{thm:trotter_sample_complexity,thm:random_sample_complexity,cor:trotter_rmse,cor:rte_rmse}). We present results for two different strategies to implement Hamiltonian simulation. The two approaches are (1) product formulas~(\PF) \cite{trot_theory,Childs:2018xxl} and (2) an implementation based on \emph{random Taylor expansion} (\RTE)~\cite{Wan:2021non,Wang:2023els}.
    \item Numerical demonstrations provided for a set of relevant small-scale examples that empirically confirm the validity of our results while also questioning the practicality of such randomized schemes.
\end{enumerate}

We make explicit which elements of our work are inherited and which are new. From prior work
we take: the integral representation of the inverse function
\eqref{eq:inverse_function_double_integral}, used in~\cite{AndrewM:2015yvx}; the
Gauss--Legendre--trapezoidal discretization of that representation, proposed
in~\cite{qrt24}; the randomized Fourier-sampling framework---sampling Fourier
times, applying a Hadamard test, and averaging---developed
in~\cite{Wan:2021non, Wang:2023els}; and the random Taylor expansion (random
compiler) lemma and its sampling procedure from~\cite{Wan:2021non}. The genuinely
new contributions of this work are: (i) \emph{explicit, non-asymptotic} choices
for every parameter of the Fourier construction
given a target error, which prior treatments left
asymptotic or unspecified; (ii) explicit non-asymptotic bounds on the sampling
normalization and the resulting sample complexity; (iii)
resource analysis of two distinct Hamiltonian-simulation kernels---a second-order
product formula and the random Taylor expansion---in both high-probability and
RMSE convergence modes, including the bias-amplification mechanism that fixes the
per-sample simulation cost; and (iv) numerical demonstrations of the
two kernels on matched matrix instances.

\section{Results}
\label{sec:results}

\subsection{Problem statement}

Let $\mathcal{H}$ be an $N = 2^n$-dimensional Hilbert space and $A \in \mathbb{C}^{N \times N}$: $\mathcal{H} \rightarrow \mathcal{H}$ be some Hermitian operator on this space. We assume that the Pauli decomposition of $A$ is known and given by,
\begin{equation}
A = \sum_{\ell \in [L]} c_{\ell}P_{\ell},
\label{eq:paulidecomp}
\end{equation} 
where $[L] := \{0, 1, \dots, L-1\}$, each $P_{\ell} \in \{I,X,Y,Z\}^{\otimes n}$ (i.e., a length-$n$ Pauli string), and each coefficient $c_{\ell} \in \mathbb{R}$. We assume efficient classical access to the coefficients and classical representation of the Paulis, $\{ (c_{\ell}, P_\ell) \}_{\ell \in [L]}$. Each length-$n$ Pauli string requires $2n$ classical bits of storage and each coefficient is assumed to be stored in some floating point representation requiring a constant amount of storage per coefficient. We also assume access to the Pauli weight of $A$, defined as  $\lambda := \sum_{\ell} |c_\ell|$. Furthermore, we assume that $\lVert A \rVert \geq 1$. From the unitarity of the Pauli matrices and the triangle inequality, this automatically implies that $\lambda \geq 1$. We also assume that the condition number of $A$, $\kappa := \kappa(A) = \lVert A^{-1} \rVert \lVert A \rVert$, is known or that we have an upper bound $\kappa^* \geq \kappa$. 
Finally, we also assume access to state preparation oracles for the pair of states $\ket{\psi},\ket{\phi} \in \mathcal{H}$, which means that we have unitaries $U_{\psi},U_{\phi}$ such that for a reference state $\ket{g} \in \mathcal{H}$,
\begin{align}
    \ket{\psi} = U_{\psi}\ket{g}, \hspace{1mm}  \ket{\phi} = U_{\phi}\ket{g}.
\label{eq:stateprep}
\end{align}
For the remainder of this paper, we assume that the cost associated with state preparation is constant and, for simplicity, we keep $\ket{g} = \ket{0}$. The problem statement that we study in this paper can be summarized as follows:

\begin{problem}
\label{problem_statement}
Given classical access to the Pauli decomposition \eqref{eq:paulidecomp} for an $N \times N$ Hermitian operator $A$ with Pauli weight $\lambda$, condition number upper bounded by $\kappa^*$, state preparation oracles $U_{\psi},U_{\phi}$ for $\ket{\psi}, \ket{\phi}$, and approximation error $\epsilon > 0$, construct an estimator $\hat{Z}$ where $\EX[\hat{Z}] \approx \bra{\phi} A^{-1} \ket{\psi} $ . The number of samples $N_{\rm S}$ is to be chosen such that the statistical average $Z_{N_{\rm S}} := N_{\rm S}^{-1}\sum_{i\in[N_{\rm S}]} \hat{Z}_i$ satisfies
\begin{equation}
\left| Z_{N_{\rm S}} - \bra{\phi} A^{-1} \ket{\psi} \right| \leq \epsilon,
\label{eq:problem}
\end{equation}
with probability greater than $1 - \delta$ for $\delta \in (0, 1)$.
\end{problem}

\subsection{Description of the randomized algorithm}
\label{subsec:alg_description}

The following provides an overview of the different steps of the algorithm with detailed results presented in the following sections. We highlight whether each individual step requires classical or quantum resources (or both). 

The algorithm begins with a series of \textbf{classical pre-processing} steps, the first of which is to approximate 
the inverse function, $f(x) = x^{-1}$, over some domain $\domain \subset \RR_0$
by a Fourier series approximation $F(x)$,
\begin{equation}
    \label{eq:fourier_series_def}
    \restr{f(x)}{\domain} \approx \sum_{s \in [S]} \alpha_s e^{-\imi x t_s} =: F(x),
\end{equation}
that satisfies,
\begin{equation}
\sup_{x \in \domain} \left| f(x) - F(x) \right| \leq \epsilon_F,
\label{eq:scalar_approx}
\end{equation}
where $S$ is the number of terms in the Fourier series approximation, $t_s \in \mathbb{R}$ is the Fourier time, $\alpha_s \in \mathbb{C}$ is the Fourier coefficient, and $\epsilon_F$ is the error on the Fourier series approximation. If the spectral range of the operator $A$ falls within the domain $\domain$, then \cref{eq:scalar_approx} implies that
\begin{equation}
\sup_{x \in \domain}
\lVert A^{-1} - F(A) \rVert \leq \epsilon_F.
\label{eq:op_approx}
\end{equation}
Given $\kappa^*$ and $\lambda$, we show how to construct $\domain$ and then $F(A)$ given an error tolerance $\epsilon_F$. We show that the complexity of this step is independent of the Hilbert space dimension.

After this initial pre-processing, we use a combination of \textbf{classical and quantum processing}. In particular, we use an iterative and randomized routine such that at each step $j$ we sample a Fourier time $t_j \in \{t_s\}_{s\in[S]}$ with probability $|\alpha_j|/\sum_{s\in[S]}|\alpha_s|$ and consider the exponential $e^{-\imi At_j}$. We now seek a method to implement an approximate controlled version of the unitary to derive an estimate for $\bra{\phi} e^{-\imi At_j} \ket{\psi}$ using a Hadamard test. At this juncture, we consider two approaches for constructing a unitary that serves as either an approximation or estimator (respectively) for the controlled version of $e^{-\imi At_j}$ that we use for implementing the Hadamard test circuit:

\begin{itemize}

\item \textit{Product Formula} (\PF): We can approximate $e^{-\imi At_j}$ with a product/Trotter formula by splitting the full exponential into products of ``simpler'' exponentials such that each exponential involves only a single $n$-qubit Pauli string. For example, one can use a $2^{\mathrm{nd}}$ order \PF\ approximation as follows:
\begin{align}
    e^{-\imi At_j} \approx U(t_j) := e^{-\imi c_0P_0t_j / 2}e^{-\imi c_1P_1t_j / 2}\dots e^{-\imi c_{L-1}P_{L-1}t_j}e^{-\imi c_{L-2}P_{L-2}t_j / 2}\dots e^{-\imi c_0P_0t_j / 2}.
\end{align}
We note that to implement a controlled version of $U(t_j)$ on quantum hardware, we can control each of the ``simpler'' exponentials individually.

\item \textit{Random Taylor Expansion} (\RTE): Alternatively, we can use Lemma 2 from \cite{Wan:2021non} to approximately express the time evolution operator $e^{-\imi At_j}$ as a linear combination of unitaries (LCU):
\begin{align}
    e^{-\imi At_j} \approx \sum_{m\in\mathcal{M}}d_{mj}U_{mj},
\end{align}
for some appropriately defined index set $\mathcal{M}$. From this LCU we sample an individual unitary $U(t_j)\in\{U_{mj}\}_{m\in\mathcal{M}}$ (by interpreting the coefficients in the LCU as being proportional to probabilities) such that $\EX[U(t_j)] \approx e^{-\imi At_j}$, thus introducing a second level of randomization. We show that for appropriately chosen algorithmic parameters the unitary $U(t_j)$ can be expressed as a series of Pauli rotations and an insignificant number of Clifford operations. Thus, the primary source of the gate complexity to construct a controlled call to $U(t_j)$ on quantum hardware comes from controlling a series of Pauli rotations.
\end{itemize}

For each Fourier sample, let $U(t_j)$ be the unitary obtained from one of the two subroutines outlined above. We define $N_{\rm CP}$ as the number of controlled Pauli rotations required to implement a controlled call to $U(t_j)$ and discuss how to choose this parameter given a certain error budget. We set $\tilde{U}(t_j) := U_{\phi}^{\dagger} U(t_j) U_{\psi}$. We now use \emph{two shots} (see Supplementary Information~\ref{app:1_shot_Hadamard} for why we choose this method over performing multiple shots for each Fourier sample), one each for the real part and imaginary part of the overlap, of the Hadamard test to obtain estimates for the overlap $\bra{0}\tilde{U}(t_j)\ket{0} = \bra{\phi} U(t_j) \ket{\psi}$. We denote the result as $\hat Z(t_j) \in \mathbb{C}$, where we note that $\EX[\hat Z(t_j)] = \bra{\phi} U(t_j) \ket{\psi}$. This completes a single iteration $j$ of this hybrid workflow. We denote the total number of iterations, and thus the total number of samples of the form $\hat Z(t_j)$, by $N_{\rm S}$. Each such $\hat Z(t_j)$ is then stored on the classical device. We note that all such iterations are independent of each other and can be conveniently parallelized.

The final step involves \textbf{classical post-processing} to compute the statistical average of all the samples that are collected:  
 \begin{align}
    \label{eq:z_hat}
     Z_{N_{\rm S}} = N_{\rm S}^{-1} \sum_{j \in [N_{\rm S}]} \hat Z(t_j).
 \end{align}
The quantities $N_{\rm S}$ and $N_{\rm CP}$ are chosen such that
\begin{align}
\label{eq:error_ineq}
    \left| Z_{N_{\rm S}} - \bra{\phi}A^{-1}\ket{\psi} \right| &\leq \left| Z_{N_{\rm S}} -  \bra{\phi}
    F(A)
    \ket{\psi} \right| + \left|\bra{\phi}
    F(A)
    \ket{\psi} - \bra{\phi}A^{-1}\ket{\psi} \right|
    \leq \epsilon_S + \epsilon_F \leq \epsilon,
\end{align}
with probability greater than $1-\delta$, thus solving \cref{problem_statement}.

\cref{eq:error_ineq} decomposes the total approximation error into two sources: the sampling error $\epsilon_S$ from estimating $\bra{\phi} F(A) \ket{\psi}$, and the Fourier series approximation error $\epsilon_F$. Each of these can be further attributed to more detailed sources of error, which we control independently throughout our analysis using the balanced budget split $\epsilon_S, \epsilon_F \leq \epsilon/2$. The key \emph{figures of merit} (or cost model) we consider are the sample
complexity $N_{\rm S}$ and the number of controlled Pauli rotations per sample
$N_{\rm CP}$, where $N_{\rm CP}$ does not include the cost of preparing
$\ket{\psi}$ and $\ket{\phi}$. We assume access to state-preparation oracles
$U_\psi, U_\phi$ at unit cost throughout; if one instead wishes to account for
their implementation, the corresponding gate cost enters $N_{\rm CP}$ additively
on a per-sample basis and propagates through all subsequent resource estimates
without otherwise altering the analysis, since each sample applies $U_\psi$ and
$U_\phi$ exactly once. We derive lower bounds on these figures of merit for each of the two Hamiltonian simulation subroutines, and we will see that there is a trade-off between them: in general, increasing $N_{\rm S}$ allows us to decrease $N_{\rm CP}$ and vice versa.

While our analysis throughout produces non-asymptotic bounds with explicit prefactors, it is useful to extract the asymptotic dependence on the key problem parameters for the purposes of comparison and discussion. \cref{tab:resource_summary} summarizes the asymptotic scaling of $N_{\rm S}$ and $N_{\rm CP}$ for the two Hamiltonian-simulation subroutines derived in this work, with pointers to the cited theorems and corollaries for the explicit non-asymptotic formulas. For each subroutine we consider two modes of convergence: a high-probability bound on the deviation $|Z_{N_{\rm S}} - \bra{\phi} A^{-1} \ket{\psi}|$ obtained via Hoeffding's inequality, as required by \cref{problem_statement}, and a bound on the root mean squared error (RMSE) of the same deviation, which trades the logarithmic dependence on $1/\delta$ for tighter constants. The two modes yield identical $N_{\rm CP}$ and differ in $N_{\rm S}$ only by the $\ln(1/\delta)$ factor present in the high-probability mode; the table therefore reports the high-probability entries, and the caption records the corresponding RMSE statements. The \RTE\ entries are stated for the regime $r = \mathcal{O}(t_{\max}^2)$, where $r$ denotes the number of equal segments into which each sampled Fourier time is partitioned for time-evolution simulation and $t_{\max} := \max_{s \in [S]} |t_s|$ is the largest Fourier time on the grid; this regime suppresses an exponential prefactor in $N_{\rm S}$ that would otherwise appear, at the cost of a quadratic dependence of $N_{\rm CP}$ on $t_{\max}$. The \PF\ subroutine has an analogous parameter $r$ controlling the second-order product formula construction.

The table reveals two structural features. First, for both subroutines and both convergence modes the sample complexity scales as $\widetilde{\mathcal{O}}(\kappa^2/\epsilon^2)$, where $\widetilde{\mathcal{O}}$ hides factors polylogarithmic in $\kappa$, $\lambda$, and $1/\epsilon$. Second, the per-sample non-Clifford gate complexity $N_{\rm CP}$ also scales quadratically in $\kappa$ for both subroutines, but \PF\ and \RTE\ differ in their dependence on the Pauli weight $\lambda$, the Pauli sparsity $L$, and the approximation error $\epsilon$.

\begin{table}[h]
\centering
\renewcommand{\arraystretch}{1.6}
\resizebox{\textwidth}{!}{%
\begin{tabular}{c|c|c|c|c}
\toprule
\textbf{Subroutine}
  & $\boldsymbol{N_{\rm S}}$ \textbf{(RMSE)}
  & $\boldsymbol{N_{\rm S}}$ \textbf{(High prob.)}
  & $\boldsymbol{N_{\rm CP}}$
  & \textbf{Total} $\boldsymbol{\mathcal{O}(N_{\rm S} N_{\rm CP})}$ \\
\midrule
\PF
  & $\widetilde{\mathcal{O}}\!\left(\dfrac{\kappa^2}{\epsilon^2}\right)$
  & $\widetilde{\mathcal{O}}\!\left(\dfrac{\kappa^2}{\epsilon^2} \ln\dfrac{1}{\delta}\right)$
  & $\widetilde{\mathcal{O}}\!\left(\dfrac{L \lambda^{3/2} \kappa^2}{\sqrt{\epsilon}}\right)$
  & $\widetilde{\mathcal{O}}\!\left(\dfrac{L \lambda^{3/2} \kappa^4}{\epsilon^{5/2}} \ln\dfrac{1}{\delta}\right)$ \\
\midrule
\RTE
  & $\widetilde{\mathcal{O}}\!\left(\dfrac{\kappa^2}{\epsilon^2}\right)$
  & $\widetilde{\mathcal{O}}\!\left(\dfrac{\kappa^2}{\epsilon^2} \ln\dfrac{1}{\delta}\right)$
  & $\widetilde{\mathcal{O}}\!\left(\lambda^2 \kappa^2\right)$
  & $\widetilde{\mathcal{O}}\!\left(\dfrac{\lambda^2 \kappa^4}{\epsilon^2} \ln\dfrac{1}{\delta}\right)$ \\
\bottomrule
\end{tabular}%
}
\caption{Asymptotic dependence of the sample complexity $N_{\rm S}$, the
per-sample non-Clifford gate complexity $N_{\rm CP}$, and the total
non-Clifford gate complexity $\mathcal{O}(N_{\rm S} N_{\rm CP})$ on the
condition number upper bound $\kappa^*$, the Pauli weight $\lambda$, the Pauli
sparsity $L$, the error tolerance $\epsilon$, and the failure probability
$\delta$, extracted from the non-asymptotic bounds shown in \cref{thm:trotter_sample_complexity,thm:random_sample_complexity,cor:trotter_rmse,cor:rte_rmse}. The $\widetilde{\mathcal{O}}$ notation hides factors
polylogarithmic in $\kappa$, $\lambda$, and $1/\epsilon$. The sample complexity
is shown in two convergence modes: a \emph{root-mean-squared-error} (RMSE)
bound (\PF: \cref{cor:trotter_rmse}; \RTE: \cref{cor:rte_rmse}), and a
\emph{high-probability} bound based on Hoeffding's inequality
(\PF: \cref{thm:trotter_sample_complexity}; \RTE:
\cref{thm:random_sample_complexity}); the latter carries an additional
$\ln(1/\delta)$ factor in exchange for a stronger mode of convergence. The
$N_{\rm CP}$ column is common to both modes. The \textbf{Total} column reports
the product $N_{\rm S} N_{\rm CP}$ in the high-probability mode. \RTE\ entries
assume the segment number is chosen as $r = \mathcal{O}(t_{\max}^2)$ to suppress
the exponential prefactor $e^{2 t_{\max}^2/r}$ that would otherwise appear in
$N_{\rm S}$; choosing $r = \mathcal{O}(t_{\max})$ instead yields $N_{\rm CP} =
\widetilde{\mathcal{O}}(\lambda \kappa)$ but introduces a factor
$e^{\mathcal{O}(\lambda\kappa\log(\lambda\kappa/\epsilon_T))}$ in $N_{\rm S}$.
Explicit non-asymptotic formulas, including all prefactors and polylogarithmic
factors, are given in the cited theorems and corollaries.}
\label{tab:resource_summary}
\end{table}

Multiplying the two figures of merit gives the total non-Clifford gate
complexity $\mathcal{O}(N_{\rm S} N_{\rm CP})$, which evaluates to
$\widetilde{\mathcal{O}}(L \lambda^{3/2} \kappa^4 / \epsilon^{5/2})$ for \PF\
and $\widetilde{\mathcal{O}}(\lambda^2 \kappa^4 / \epsilon^2)$ for \RTE:
both subroutines scale as $\kappa^4$, arising as the product of the quadratic
sampling complexity $N_{\rm S} = \widetilde{\mathcal{O}}(\kappa^2)$ and a
quadratic per-sample simulation cost $N_{\rm CP} = \widetilde{\mathcal{O}}(\kappa^2)$. This $\kappa^4$
scaling, together with the large $\kappa$-dependent prefactors entering
$N_y N_z$ and $t_{\max}$, is ultimately what renders the approach impractical:
as we show in our numerical results, even a $4 \times 4$ instance with
$\kappa = 100$ and target error $\epsilon \sim 10^{-3}$ already requires on the
order of $10^{15}$ controlled Pauli rotations to converge.

We provide a high-level flow chart of the full randomized quantum algorithm in \cref{fig:alg_outline} that summarizes the content of this section, and we numerically study the \RTE\ and \PF\ methods by performing classical simulations for a selection of small-scale examples. We note that these methods are not specific to the inverse problem and can be used to estimate scalar properties of any matrix function $f(A)$ for which a provably accurate Fourier series approximation can be constructed.

\begin{figure}[h]
    \centering
    \input{outline_figure}
    \caption{Flowchart outlining the randomized quantum algorithm for generating approximate solutions to \cref{problem_statement}. Orange boxes are classical data processing steps, with boxes having a white background describing inputs, intermediate results, and outputs, and those filled with an orange background describing a classical algorithm. The blue box shows the Hadamard test circuit which is the only step that requires a QPU. The dashed lines indicate the iterative portion of the algorithm that uses a combination of classical and quantum processing. The dashed arrows indicate a choice in the method used to perform Hamiltonian simulation.}
    \label{fig:alg_outline}
\end{figure}

\subsection{Fourier series approximation to the inverse function with explicit error bounds}
\label{sec:fourier}

We start with the following integral representation for the inverse function, which was also used in~\cite{AndrewM:2015yvx},
\begin{equation}
    \label{eq:inverse_function_double_integral}
    x^{-1} = \frac{\imi}{\sqrt{2 \pi}} \int_{0}^{\infty} dy \int_{- \infty}^{\infty} dz \left( z e^{-z^2/2}e^{-\imi xyz} \right) =: I(x).
\end{equation}
We define the domain $\domain_\kappa = \left[-1, -1/\kappa \right] \cup \left[1/\kappa, 1 \right]$, with $\kappa \geq 1$, which excludes the singularity at zero, and construct a finite linear combination of unitaries (LCU) approximation---equivalently, a Fourier series approximation---to $I(x)$ over $\domain_\kappa$ in two steps: truncating the integration ranges in \cref{eq:inverse_function_double_integral}, and then discretizing the truncated integral. We improve upon the construction in~\cite{AndrewM:2015yvx} by adopting a more efficient discretization, following~\cite{qrt24}. We make these two steps precise via the following definitions.

\begin{definition}[$\epsilon_T$-close integral truncation]
\label{def:integral_truncation}
For positive real \emph{truncation parameters} $(y_{\max}, z_{\max}) \in \RR^+ \times \RR^+$, the truncated integral $I_T: \domain_\kappa \rightarrow \mathbb{C}$ is defined as
\begin{equation}
\label{eq:inverse_function_truncated_double_integral}
    I_T(x) := \frac{\imi}{\sqrt{2\pi}} \int_{0}^{y_{\max}} dy \int_{-z_{\max}}^{z_{\max}} dz \left( z e^{-z^2/2}e^{-\imi xyz} \right).
\end{equation}
The truncation is said to be $\epsilon_T$-close to $I$ if the parameters $(y_{\max}, z_{\max})$ satisfy
\begin{equation}
    \label{eq:truncation_error}
    \sup_{x \in \domain_\kappa} \left| I_T(x) - I(x) \right| \leq \epsilon_T,
\end{equation}
in which case we write $I_T^{(\epsilon_T)}$ to make this dependence explicit.
\end{definition}

We next discretize the truncated integral. This requires positive \emph{weight functions} $w_y: [0, y_{\max}] \rightarrow \RR^+$ and $w_z: [-z_{\max}, z_{\max}] \rightarrow \RR^+$, together with quadrature nodes $y_j \in [0, y_{\max}]$ and $z_k \in [-z_{\max}, z_{\max}]$.

\begin{definition}[$\epsilon_D$-close integral discretization]
\label{def:integral_discretization}
For \emph{Fourier parameters} $(J, K) \in \mathbb{N} \times \mathbb{N}$, the discretization $I_D: \domain_\kappa \rightarrow \mathbb{C}$ of $I_T$ is defined as
\begin{equation}
\label{eq:fourier_approx_def}
I_D(x) := \frac{\imi}{\sqrt{2\pi}} \sum_{j \in [J]} \sum_{k \in [K]} w_y^{(j)} w_z^{(k)} z_k e^{-z_k^2/2} e^{-\imi xy_jz_k},
\end{equation}
where $w_y^{(j)} := w_y(y_j)$ and $w_z^{(k)} := w_z(z_k)$. The discretization is said to be $\epsilon_D$-close to $I_T$ if the parameters $(J, K)$ satisfy
\begin{equation}
    \label{eq:discretization_error}
    \sup_{x \in \domain_\kappa} \left| I_D(x) - I_T(x) \right| \leq \epsilon_D,
\end{equation}
in which case we write $I_D^{(\epsilon_D, \epsilon_T)}$ to make explicit both the discretization tolerance and the underlying truncation tolerance of $I_T$.
\end{definition}

The discretization in \cref{eq:fourier_approx_def} is a Fourier series approximation to the inverse function with $JK$ terms: each product $y_j z_k$ plays the role of a \emph{Fourier time}, and each $\imi w_y^{(j)} w_z^{(k)} z_k e^{-z_k^2/2}/\sqrt{2\pi}$ is the corresponding (complex) \emph{Fourier coefficient}.

Combining \cref{def:integral_truncation,def:integral_discretization}, the total error in approximating $I(x)$ by $I_D^{(\epsilon_D, \epsilon_T)}(x)$, which we denote by $\epsilon_F$ in accordance with \cref{eq:error_ineq}, decomposes via the triangle inequality into truncation and discretization contributions:
\begin{align*}
    \epsilon_F := \sup_{x \in \domain_\kappa} \left| I_D^{(\epsilon_D, \epsilon_T)}(x) - I(x) \right|
    &\leq \sup_{x \in \domain_\kappa} \left| I_D^{(\epsilon_D, \epsilon_T)}(x) - I_T^{(\epsilon_T)}(x) \right| + \sup_{x \in \domain_\kappa} \left| I_T^{(\epsilon_T)}(x) - I(x) \right| \\
    &\leq \epsilon_D + \epsilon_T.
\end{align*}
Our goal in the remainder of this section is to construct a provably accurate Fourier series approximation that requires as few terms $JK$ and as short a maximum simulation time $t_{\max} := y_{\max} z_{\max}$ as possible, given a target approximation error $\epsilon$; the results are summarized in \cref{alg:fourier_series_construction}. To meet the balanced choice $\epsilon_F \leq \epsilon/2$ made earlier, we apply \cref{alg:fourier_series_construction} with parameters chosen so that $\epsilon_D + \epsilon_T \leq \epsilon/2$. 

The following theorem characterizes the truncation error $\epsilon_T$ and describes how to choose the truncation parameters $(y_{\max},z_{\max})$ for each of the integrals appearing in \cref{eq:inverse_function_truncated_double_integral}.

\begin{theorem}
\label{thm:truncation_parameters}
Given a truncation error $\epsilon_T > 0$, the integration bounds $(y_{\max}, z_{\max}) \in \mathbb{R}^{+} \times \mathbb{R}^{+}$ given by 
\begin{align}
y_{\max} = \kappa \sqrt{2 \ln{\left(\frac{3\kappa}{\epsilon_T}\right)}}, \hspace{2mm}
z_{\max} = \sqrt{2 \ln{\left(\frac{3\kappa}{\epsilon_T}\right)}}
\label{eq:trunc_params}
\end{align}
provide a truncated integral approximation $I_T^{(\epsilon_T)}$, as defined in~\cref{def:integral_truncation}, to $I(x)$ over the domain $\domain_{\kappa}$.
\end{theorem}
\begin{proof}
    The proof is given in Supplementary Information~\ref{app:thm3}
\end{proof}
We note that \cref{thm:truncation_parameters}  allows us to compute the maximal Fourier time parameter $t_{\max}$ for the Fourier series $I_T^{(\epsilon_T)}$ as follows,
\begin{equation}
\label{eq:max_fourier_time}
t_{\max} = 2\kappa\ln{\left(\frac{3\kappa}{\epsilon_T}\right)}.
\end{equation}

Next, we present lower bounds on the Fourier parameters $(J,K)$ that appear in each of the sums in \cref{eq:fourier_approx_def} such that the discretization error to approximate an integral of the form $I_T^{(\epsilon_T)}$ is bounded from above by a given error tolerance $\epsilon_D$. We note that this analysis requires specific choices for the quadrature rules for both the $z$- and $y$-integral appearing in \cref{eq:inverse_function_truncated_double_integral}. We follow the proposal introduced in~\cite{qrt24} and combine discretization by the trapezoidal rule for the $z$-integral and the Gauss-Legendre rule for the $y$-integral to select the quadrature nodes and weights. 

For the $z$-integral appearing in \cref{eq:inverse_function_truncated_double_integral}, we observe that the integrand is quasi-periodic, rapidly decaying (Gaussian), and analytic. It has been established that under these conditions the trapezoidal rule  convergences exponentially~\cite{trefethen14}. In the notation of \cref{eq:fourier_approx_def}, for an appropriately chosen value of the Fourier parameter $K$ (given some error tolerance), we can discretize the $z$-integral using the trapezoidal rule with the constant weight function,
\begin{align}
    \label{eq:trap_weight}
    w_z^{(k)} = \frac{2z_{\max}}{(K-1)}  =: \Delta z,
\end{align}
with quadrature nodes $z_k = \Delta z (k - (K-1)/2)$ for $k \in [K]$. For the $y$-integral appearing in \cref{eq:inverse_function_truncated_double_integral}, the integrand is a purely oscillatory function that does not decay and hence the trapezoidal rule does not converge exponentially fast. As proposed in~\cite{qrt24}, we instead make use of the Gauss-Legendre quadrature rule for the $y$-integral. We use the change of variable $y \mapsto y_{\max}(1+y)/2$ to rescale and write \cref{eq:inverse_function_truncated_double_integral} as,
\begin{equation}
    \label{eq:GL_transformed_integral}
    I_T^{(\epsilon)}(x) = \frac{\imi y_{\max}}{2\sqrt{2 \pi}} \int_{-1}^{1} dy \int_{- z_{\max}}^{z_{\max}} dz \left( z e^{-z^2/2}e^{-\imi xy_{\mathrm{max}}(1+y)z/2} \right).
\end{equation}
In the notation of \cref{eq:fourier_approx_def}, we can thus discretize the $y$-integral in \cref{eq:GL_transformed_integral} by using $J$ many Legendre nodes and setting $y_j = y_{\max}(1+c_j)/2$, where $c_j \in [-1,1]$ is the $j^{\mathrm{th}}$ root of the $J^{\mathrm{th}}$ degree Legendre polynomial. We also set $w_y^{(j)}$ to be the usual Gauss-Legendre quadrature weight scaled by the factor $y_{\max}/2$, i.e.,
\begin{equation}
    \label{eq:gl_weight}
    w_y^{(j)} = \frac{y_{\max}}{2}\frac{2}{(1-y_j^2)(P_J''(y_j))^2},
\end{equation}
where $P_J$ denotes the degree-$J$ Legendre polynomial.

Thus, we have effectively approximated the integral appearing in \cref{eq:inverse_function_truncated_double_integral} using a discretization of the form appearing in \cref{eq:fourier_approx_def}. We use the error bounds presented in \cite{trefethen14,qrt24} to arrive at the following theorem on how to choose the Fourier parameters $(J,K)$ to control the total discretization error.

\begin{theorem}
\label{thm:fourier_parameters}
Given a discretization error $\epsilon_D > 0$, the Fourier parameters $(J, K) \in \mathbb{N} \times \mathbb{N}$,
\begin{align}
K &= 1 + \frac{1}{\pi} \left( \ln{(2)} + \frac{z_{\max}^2}{2} + 2\kappa\ln\left(\frac{3\kappa}{\epsilon_T}\right) + \ln{\left(\frac{2}{\epsilon_D}\right)} + \ln{\left(1 + \frac{2}{z_{\max}\sqrt{2\pi}}\right)} \right), \label{eq:fpK}\\
J &= \frac{1}{\ln(2)} \left( \ln\left(\frac{K}{K-1}\right) + \ln\left(\frac{32 y_{\max} z_{\max}^2}{\epsilon_D \sqrt{2\pi}}\right) + \frac{\kappa \ln\left(\frac{3\kappa}{\epsilon_T}\right)}{2\sqrt{2}} \right), \label{eq:fpJ}
\end{align}
provide an integral discretization $I_D^{(\epsilon_D, \epsilon_T)}$, as defined in \cref{def:integral_discretization}, of $I_T^{(\epsilon_T)}(x)$ over the domain $\domain_\kappa$.
\end{theorem}
\begin{proof}
    The proof is given in Supplementary Information~\ref{app:thm4}
\end{proof}

We can apply the results presented in this section to the task of approximating the inverse of a Hermitian matrix. We remind the reader that we assume $\lVert A\rVert \geq 1$  and that we know an upper bound on the condition number $\kappa \leq \kappa^*$. We distinguish between two cases. If $\lVert A\rVert=1$, then $\Lambda(A) \subseteq \domain_{\kappa^*}$, where $\Lambda(A)$ indicates the set of eigenvalues of $A$, and the previous results readily apply. Alternatively, if $\lVert A\rVert > 1$, then it does not necessarily hold that $\Lambda(A) \subseteq \domain_{\kappa^*}$. In this case, we can rescale the problem as follows.

\begin{lemma}
    \label{lemma:mat_rescaling}
    Let A be a Hermitian matrix with $\lVert A\rVert > 1$, upper bound on the condition number $\kappa^*$, and Pauli weight $\lambda$. Define the rescaled matrix $\Tilde{A}$ and the rescaled upper bound on the condition number $\Tilde{\kappa}^*$ as follows:
    \begin{align*}
        \Tilde{A} &:= \lambda^{-1}A, & \Tilde{\kappa}^* &:= \lambda \kappa^*.
    \end{align*}
    Then $\Lambda(\Tilde{A}) \subseteq \domain_{\Tilde{\kappa}^*}$.
\end{lemma}
\begin{proof}
    The proof is given in Supplementary Information~\ref{app:matrix_rescaling}.
\end{proof}
Note that the Pauli weight of $\Tilde{A} = 1$. Thus, we may apply \cref{thm:truncation_parameters,thm:fourier_parameters} to $\Tilde{A}$ on the domain $\domain_{\tilde{\kappa}^*}$ to obtain a Fourier series approximation for $\Tilde{A}^{-1}$. We can then multiply this Fourier series by $\lambda^{-1}$ to obtain an approximation of $A^{-1}$. Note that the error on the approximation to $A^{-1}$ is also scaled by $\lambda^{-1}$.

In the rest of this work we present our results in terms of the rescaled upper bound on the condition number $\Tilde{\kappa}^*$. \cref{alg:fourier_series_construction} summarizes the content of this section and combines it with the previous results from this section.

\begin{algorithm}[h]
\caption{\label{alg:fourier_series_construction} Fourier series approximation to the inverse function $A^{-1}$.}
    \KwIn{\begin{itemize}
        \item[-] $\kappa^*$: upper bound on the condition number of $A$,
        \item[-] $\lambda$: Pauli weight of $A$,
        \item[-] $\epsilon_T$, $\epsilon_D$ : truncation and discretization error tolerances respectively such that $\epsilon_T + \epsilon_D \leq \epsilon/2$.
    \end{itemize}}
    \KwOut{Fourier coefficients and time parameters $\{\alpha_s,t_s\}_{s \in [S]}$ such that $\lVert A^{-1} - \lambda^{-1}\sum_{s \in [S]} \alpha_se^{-\imi \Tilde{A} t_s} \rVert \leq \frac{(\epsilon_T + \epsilon_D)}{\lambda}$}

$\Tilde{\kappa}^* \gets \lambda\kappa^*$\;
$y_{\max} \gets \Tilde{\kappa}^*\sqrt{2 \ln{\left(\frac{3\Tilde{\kappa}^*}{\epsilon_T}\right)}}$, $z_{\max} \leftarrow \sqrt{2 \ln{\left(\frac{3\Tilde{\kappa}^*}{\epsilon_T}\right)}}$\tcp*{\cref{eq:trunc_params}}
$K \gets 1 + \frac{1}{\pi} \left( \ln{(2)} + \frac{z_{\max}^2}{2} + 2\Tilde{\kappa}^*\ln\left(\frac{3\Tilde{\kappa}^*}{\epsilon_T}\right) + \ln{\left(\frac{2}{\epsilon_D}\right)} + \ln{\left(1 + \frac{2}{z_{\max}\sqrt{2\pi}}\right)} \right)$\tcp*{\cref{eq:fpK}}
$J \gets \frac{1}{\ln\left(2\right)} \left(\ln\left(\frac{K}{K-1}\right) + \ln\left(\frac{32y_{\max}z_{\max}^2}{\epsilon_D\sqrt{2\pi}}\right) + \frac{\Tilde{\kappa}^*\ln\left(\frac{3\Tilde{\kappa}^*}{\epsilon_T}\right)}{2\sqrt{2}} \right)$\tcp*{\cref{eq:fpJ}}
$\Delta z \gets \frac{2z_{\max}}{K-1}$\;
\ForEach{$(j,k) \in [J]\times[K]$}{
$x_j \gets$ $j$-th root of the $J$-th degree Legendre polynomial\;
$y_j \gets y_{\max}(1+x_j)/2$\;
$w_y^{(j)} \gets \frac{y_{\max}}{2}\frac{2}{(1-y_j^2)(P_J''(y_j))^2}$ \tcp*{\cref{eq:gl_weight}}
$z_k \gets \Delta z \left(k - (K-1)/2 \right)$\;
$\alpha_{jk} \leftarrow \left(\frac{\imi}{\sqrt{2 \pi}} w_y^{(j)} \Delta z\right) z_k e^{-z_k^2/2}$\;
$t_{jk} \gets y_jz_k$\;
}
\KwRet{$\{\alpha_{jk}, t_{jk}\}_{(j,k) \in [J] \times [K]}$}\;
\end{algorithm}

\subsection{Sampling from the Fourier series approximation}
\label{sec:sampling}

Having characterized the Fourier series approximation error $\epsilon_F$ appearing in \cref{eq:error_ineq}, 
we shift our attention over the next two sections to characterizing the sampling error $\epsilon_S$.
In the current section, we introduce the routine we use to sample from the Fourier series approximation. To set up sampling access and interpret the Fourier coefficients as probabilities, we have to normalize this set of coefficients. In particular, we define the following normalization parameters:
\begin{align}
    N_y := \frac{1}{\sqrt{2\pi}} \sum_{j\in[J]}\left|w_y^{(j)}\right|, \hspace{2mm}
    N_z := \sum_{k\in[K]}\left|w_z^{(k)}z_ke^{-z_k^2/2}\right|,
\label{eq:normalization_parameters}
\end{align}
    where $w_y^{(j)}$ is given by \cref{eq:gl_weight} and $w_z^{(k)}$ is given by \cref{eq:trap_weight}. Since the Gauss--Legendre weights $w_y^{(j)}$ and the constant trapezoidal weight $w_z^{(k)} = \Delta z$ are positive, the modulus of the Fourier coefficient $\alpha_{jk}$ factorizes across $j$ and $k$, and we define probability distributions
\begin{align}
    p_y^{(j)} := \frac{w_y^{(j)}}{\sum_{j' \in [J]} w_y^{(j')}} = \frac{w_y^{(j)}}{\sqrt{2\pi}\, N_y}, \hspace{2mm}
    p_z^{(k)} := \frac{w_z^{(k)} |z_k| e^{-z_k^2/2}}{N_z},
    \label{eq:probs}
\end{align}
each summing to unity. We rewrite the Fourier series generated by \cref{alg:fourier_series_construction} as follows:
\begin{align}
    \label{eq:fourier_series_}
    \sum_{s\in[S]} \alpha_s e^{-\imi \tilde{A} t_s} 
    &= I_{D}^{(\epsilon_D,\epsilon_T)}(\tilde{A}) \notag \\ 
    &= N_y N_z \sum_{j\in[J]} \sum_{k\in[K]} \omega(j,k)\, p_y^{(j)} p_z^{(k)}\, e^{-\imi\tilde{A}t_{jk}}.
\end{align}
where each $\omega(j,k)$ is a unit-modulus phase factor,
\begin{equation}
    \omega(j,k) := \imi \cdot \mathrm{sign}(z_k).
\end{equation}
We can sample indices $j' \in [J]$ with probability $p_y^{(j')}$ and indices $k' \in [K]$ with probability $p_z^{(k')}$, and thus effectively sample exponentials of the form $e^{-\imi\tilde{A}t_{j'k'}}$ from \cref{eq:fourier_series_}. As we will show, the product $N_y N_z / \lambda$ controls the magnitude of the per-sample estimators constructed in the following sections and therefore governs the sample complexity of the overall algorithm. The following Lemma uses \cref{thm:truncation_parameters} to express $N_y$ exactly and to upper bound $N_z$ in terms of parameters introduced previously:
\begin{lemma}
    \label{lemma:normalization_bounds}
    The normalization parameters $N_y, N_z$ defined in \cref{eq:normalization_parameters} satisfy:
    \begin{align}
        N_y = \frac{y_{\mathrm{max}}}{\sqrt{2\pi}} = \frac{\Tilde{\kappa}^*}{\sqrt{2\pi}} \sqrt{2\ln\left(\frac{3\Tilde{\kappa}^*}{\epsilon_T}\right)}, \hspace{2mm}
        N_z \leq z_{\mathrm{max}}\sqrt{2\pi} = \sqrt{2\pi}\sqrt{2\ln\left(\frac{3\Tilde{\kappa}^*}{\epsilon_T}\right)}.
    \end{align}
\end{lemma}
\begin{proof}
    The proof is given in Supplementary Information~\ref{app:normalization_bounds}
\end{proof}
\cref{lemma:normalization_bounds} shows that $N_y$ scales linearly with $\tilde{\kappa}^*$ up to logarithmic factors, while $N_z$ is poly-logarithmic in $\tilde{\kappa}^*$. The product $N_y N_z$ therefore scales linearly with $\tilde{\kappa}^*$ up to logarithmic factors. We show that this translates to a sample complexity that depends polynomially on $\tilde{\kappa}^*$ (\cref{thm:trotter_sample_complexity,thm:random_sample_complexity}).

\subsection{Error bounds for the product formula subroutine}
\label{sec:resource_req_pf}

Let $\tau$ be a Fourier time of the form $t_{j'k'}$ that we sample using the established protocol. We now consider the task of implementing controlled versions of operators of the form $e^{-\imi\tilde{A}\tau}$. In particular, we consider two subroutines to accomplish this task: (i) a product formula (or \PF) approach to construct approximate time evolution operators, and (ii) the random Taylor expansion (or \RTE) scheme to sample from the LCU decomposition for controlled time evolution operators, which introduces a second randomized procedure into the overall algorithm. We derive the number of quantum circuit samples and gate depth per sample required to solve \cref{problem_statement} for each subroutine. This section in particular focuses on the \PF\ routine.

We consider the rescaled matrix $\tilde{A}$ and write it as,
\begin{align}
    \tilde{A} = \lambda^{-1}\sum_{\ell\in[L]}c_{\ell}P_{\ell} := \sum_{\ell\in[L]}\tilde{P}_{\ell},
\end{align}
where we have absorbed the coefficients $c_{\ell}/\lambda$ into the modified Paulis $\tilde{P}_{\ell}$. The \PF\ method allows one to construct the exponential $e^{-\imi\tilde{A}\tau}$ by decomposing the full exponential into products of exponentials each involving only a single $\tilde{P}_{\ell}$. The exponential of a single $\tilde{P}_{\ell}$ can then be implemented using the standard CNOT ladder construction~\cite{Nielsen_Chuang_2010}. The controlled version of the exponential of a single $\tilde{P}_{\ell}$ thus requires a controlled single-qubit Pauli rotation and $4\log(N)$ Cliffords~\cite{Wan:2021non}. The most common method to construct a \PF\ is using the Lie/Suzuki Trotter construction. For explicit constructions we refer the reader to \cite{trot_theory}. The number of terms in a \PF\ grows exponentially with the order of the formula, and in practice only low order product formulas are used. In the rest of this work we will consider the $2^{\mathrm{nd}}$ order formula given by,
\begin{align}
    \label{eq:2nd_order_pf}
    S(\tau) := \prod_{\ell\in[L]}e^{-\imi\tilde{P}_{\ell}\tau/2}\prod_{\ell\in[L]}e^{-\imi\tilde{P}_{L-\ell}\tau/2}.
\end{align}
We note that the following analysis generalizes to higher-order product
formulas~\cite{trot_theory}; we restrict to second order because our aim is an
explicit, non-asymptotic resource estimate, and a clean computable error prefactor
is readily available at second order, whereas the higher-order commutator bounds
carry prefactors that are not as easily evaluated in closed form for general
matrices. We leave this direction to future work.

As shown in \cite{trot_theory}, \cref{eq:2nd_order_pf} provides a good approximation to the true exponential when $\tau$ is small since the error terms scale as a polynomial in $\tau$. For finite $\tau$, we can split the full interval $\tau$ into $r$ smaller steps $\tau/r$ and use a $2^{\mathrm{nd}}$ order \PF\ approximation for each step, i.e.,
\begin{align}
    e^{-\imi \tilde{A} \tau}
    \approx \left(S(\tau/r)\right)^{r} := S_{r}(\tau).
\label{eq:pf_split}
\end{align}
The quantity $r\in\mathbb{N}$ is known as the Trotter number. The question we ask now is how to choose $r$ to ensure the approximation error,
\begin{align}
    \label{eq:trotter_error}
    \epf := \|e^{-\imi\tilde{A}\tau} - S_{r}(\tau)\|,
\end{align}
is bounded. The following lemma establishes an upper bound on $\epf$ for a given choice $r$:
\begin{lemma}
    \label{lemma:trotter_number}
    For a given Trotter number $r$, the error in operator norm $\epf$ defined in \cref{eq:trotter_error} for a $2^{\mathrm{nd}}$ order \PF\ satisfies the following upper bound:
    \begin{align}
        \label{eq:trotter_error_ub}
        \epf \leq \frac{f \tau^3}{r^2},
    \end{align}
    where $f \in \mathbb{C}$ is defined as follows:
    \begin{align}
        f := \frac{1}{12}\sum_{\ell_0 = 0}^{L-1} \left\| \left[ \sum_{\ell_2 = \ell_0}^{L-1} \tilde{P}_{\ell_2}, \left[ \sum_{\ell_1 = \ell_0+1}^{L-1} \tilde{P}_{\ell_1}, \tilde{P}_{\ell_0} \right] \right] \right\| 
        + \frac{1}{24} \sum_{\ell_0 = 0}^{L-1} \left\| \left[ \tilde{P}_{\ell_0}, \sum_{\ell_1 = \ell_0+1}^{L-1} \tilde{P}_{\ell_1} \right] \right\|.
    \end{align}
\end{lemma}
\begin{proof}
     The proof of this lemma follows from applying the triangle inequality $r$ times to the result in Proposition 16 from \cite{trot_theory}.
\end{proof}

We can use \cref{lemma:trotter_number} to construct provably accurate \PF\ approximations to exponentials of the form $e^{-\imi\tilde{A}\tau}$. In particular, given an upper bound on $\epf$, which we call $\epf*$, choosing
\begin{align}
    \label{eq:trotter_r_lb}
    r \geq \sqrt{\frac{f\tau^3}{\epf^*}},
\end{align}
ensures that $\epf \leq \epf^*$. We point out that while this provides a valid choice for the Trotter number $r$,
empirical evidence shows that in many cases of interest this bound is pessimistic and an error $\epf$ can be reached at significantly smaller Trotter numbers than that implied by \cref{eq:trotter_r_lb} (for example see \cite{Childs:2018xxl,trot_theory,Kane:2025ybw}). However, tighter analytic bounds require knowledge of the specific matrix of interest, so we stick to~\cref{eq:trotter_r_lb} for our analysis. We also note that the gate complexity, in terms of the number of single-qubit Pauli rotations, of implementing a \PF\ formula scales linearly with the total number of terms in the \PF\ and hence with $r$. In particular, we require $2rL$-many single-qubit Pauli rotations for the $2^{\mathrm{nd}}$ order \PF\ $S_r(\tau)$. We summarize the \PF\ construction routine in \cref{alg:pf}.

\begin{algorithm}[htp!]
\caption{\label{alg:pf} Construct a $2^{\mathrm{nd}}$ order \PF\ approximation to $e^{-\imi\tilde{A}\tau}$}
\KwIn{\begin{itemize}
        \item[-] Pauli decomposition $\tilde{A} = \sum_{\ell \in [L]}\tilde{P}_{\ell}$
        \item[-] $\tau \in \mathbb{R}$
        \item[-] $\epf^* \in \mathbb{R}^+$
\end{itemize}}
\KwOut{Unitary $S_{r}(\tau)$ such that:
\begin{itemize}
    \item[-] the non-Clifford cost of controlling $S_{r}(\tau)$ is equal to that of controlling \\ $2rL$ Pauli rotations 
    \item[-] $\epf \leq \epf^*$
\end{itemize}}

$U \gets \identity$

$r \gets \sqrt{f\tau^3/\epf^*}$\tcp*{\cref{eq:trotter_r_lb}}

$x \gets \tau/r$

\ForEach{$\ell \in [L]$}{
$U \gets Ue^{-i\tilde{P}_{\ell}x/2}$
}
\ForEach{$\ell \in [L]$}{
$U \gets Ue^{-i\tilde{P}_{L-\ell}x/2}$
}

$S_{\rpf}(t_{j'k'}) \gets U^{\rpf}$

\KwRet $S_{r}(\tau)$
\end{algorithm}

Having constructed $S_r(\tau)$, we estimate the overlap $\bra{\phi} S_r(\tau) \ket{\psi}$ using the Hadamard test circuit shown in \cref{fig:alg_outline} with controlled unitary $U = U_{\phi}^{\dagger} S_r(\tau) U_{\psi}$. We perform the test twice --- once for the real part and once for the imaginary part of the overlap --- yielding single-shot outcomes $\hat X_\tau, \hat Y_\tau \in \{-1, +1\}$ that are unbiased estimators of these parts conditional on $\tau$:
\begin{align}
    \mathbb{E}[\hat X_\tau \mid \tau] &= \Re \bra{0} U_{\phi}^{\dagger} S_r(\tau) U_{\psi} \ket{0}, &
    \mathbb{E}[\hat Y_\tau \mid \tau] &= \Im \bra{0} U_{\phi}^{\dagger} S_r(\tau) U_{\psi} \ket{0}.
\end{align}
Combining these with the phase factor $\omega_\tau$ and normalization $N_y N_z / \lambda$ from \cref{eq:fourier_series_}, we form the single-sample estimator
\begin{align}
    \label{eq:PF_estimator}
    \hat Z := \omega_\tau \frac{N_y N_z}{\lambda} \left( \hat X_\tau + \imi \hat Y_\tau \right),
\end{align}
where the Fourier time $\tau$ is itself sampled from the joint distribution defined in \cref{eq:probs}. The randomness in $\hat Z$ thus comprises both the sampled Fourier time and the two Hadamard outcomes; the expectation $\mathbb{E}[\hat Z]$ is taken jointly over both. We define the bias
\begin{align}
    \label{eq:trotter_bias}
    B_{\rm PF} := \frac{\bra{\phi} I_D^{(\epsilon_D, \epsilon_T)}(\tilde A) \ket{\psi}}{\lambda} - \mathbb{E}[\hat Z],
\end{align}
which captures the systematic error introduced by replacing $e^{-\imi \tilde A \tau}$ with its second-order PF approximation $S_r(\tau)$ and vanishes as $r \to \infty$. We bound this bias as follows.

\begin{theorem}
    \label{thm:trotter_bias_bound}
    The norm of the bias term $\Bpf$ appearing in \cref{eq:trotter_bias} can be upper bounded as follows:
    \begin{align}
        \label{eq:trotter_bias_ub}
        |\Bpf| \leq \lambda^{-1} N_y N_z \frac{f t_{\mathrm{max}}^3}{r^2},
    \end{align}
    where $N_y$ and $N_z$ are given by \cref{lemma:normalization_bounds}, and $t_{\mathrm{max}}$ is given by \cref{eq:max_fourier_time}.
\end{theorem}
\begin{proof}
    The proof is given in Supplementary Information~\ref{app:thm_trotter_bias_bound}.
\end{proof}
The following theorem provides a way to choose the number of samples and controlled Pauli rotations required per sample when utilizing the \PF\ subroutine:
\begin{theorem}
    \label{thm:trotter_sample_complexity}
    Let $Z_{N_{\rm S}}$ be the random variable obtained by averaging $N_{\rm S}$ independent samples of the biased estimator $\hat Z$ defined in \cref{eq:PF_estimator}, and assume the bias defined in \cref{eq:trotter_bias} satisfies $|B_{\rm PF}| < \epsilon/2$ for some choice of Trotter number $r$ (\cref{thm:trotter_bias_bound}). If $\epsilon_D + \epsilon_T \leq \epsilon/2$, then $Z_{N_{\rm S}}$ solves \cref{problem_statement} provided
    \begin{align}
        \label{eq:trotter_complexities}
        N_{\rm S} = 1 + \ln\!\left(\frac{4}{\delta}\right) \frac{16 (N_y N_z)^2}{\lambda^2 (\epsilon/2 - |B_{\rm PF}|)^2}, \qquad
        N_{\rm CP} = 2 r L,
    \end{align}
    where $N_{\rm CP}$ provides an upper bound on the number of controlled Pauli rotations required per sample and $N_y, N_z$ are given by \cref{lemma:normalization_bounds}.
\end{theorem}
\begin{proof}
The proof is given in Supplementary Information~\ref{app:thm_trotter_sample_complexity}.
\end{proof}

We observe that there exists a trade-off between $N_{\rm S}$ and $N_{\rm CP}$. For instance, if we increase the value of the Trotter number $r$, we observe that this would increase $N_{\rm CP}$. However, from \cref{thm:trotter_bias_bound}, we see that this will also decrease the upper bound on $|B_{\rm PF}|$ and hence (in general) decrease the number of samples $N_{\rm S}$. Similarly, the effect of decreasing $r$ will be to increase $N_{\rm S}$. Thus, using a more accurate \PF\ decreases the overall sample complexity while increasing the non-Clifford gate complexity and vice versa. Using \cref{thm:trotter_bias_bound}, we observe that the condition $|B_{\rm PF}| < \epsilon/2$ implies a lower bound on $r$,
\begin{equation}
    \label{eq:trotter_number_lower_bound}
    r > \sqrt{\frac{2 N_y N_z f t_{\max}^3}{\lambda \epsilon}},
\end{equation}
which yields the following lower bound on $N_{\rm CP}$ using \cref{thm:trotter_sample_complexity},
\begin{equation}
    \label{eq:ncp_lower_bound}
    N_{\rm CP} > 2L\sqrt{\frac{2 N_y N_z f t_{\max}^3}{\lambda \epsilon}}.
\end{equation}

An alternative figure of merit for a Monte Carlo estimator is the root mean squared error (RMSE), which can be bounded directly via the second moment of the estimator. The following corollary translates \cref{thm:trotter_sample_complexity} into an RMSE-based statement, yielding tighter constants in terms of $\epsilon$ at the cost of a weaker mode of convergence.

\begin{corollary}
\label{cor:trotter_rmse}
Let $Z_{N_{\rm S}}$ be the average of $N_{\rm S}$ independent samples of the biased estimator $\hat Z$ defined in \cref{eq:PF_estimator}, and let
\begin{equation}
    \mathrm{RMSE}(Z_{N_{\rm S}}) := \sqrt{\,\mathbb{E}\!\left[\,\left| Z_{N_{\rm S}} - \frac{\bra{\phi} I_D^{(\epsilon_D,\epsilon_T)}(\tilde A) \ket{\psi}}{\lambda} \right|^2\,\right]\,}.
\end{equation}
Then
\begin{equation}
    \mathrm{RMSE}(Z_{N_{\rm S}}) \;\leq\; \sqrt{\,|B_{\rm PF}|^2 + \frac{2 N_y^2 N_z^2}{\lambda^2\, N_{\rm S}}\,}.
\end{equation}
In particular, assuming the bias satisfies $|B_{\rm PF}| < \epsilon/2$ (\cref{thm:trotter_bias_bound}), the choice
\begin{equation}
    N_{\rm S} \;\geq\; \frac{8\, N_y^2 N_z^2}{\lambda^2\, ((\epsilon/2)^2 - |B_{\rm PF}|^2)}
\end{equation}
ensures $\mathrm{RMSE}(Z_{N_{\rm S}}) \leq \epsilon/2$. Combined with $\epsilon_D + \epsilon_T \leq \epsilon/2$, this guarantees an RMSE of at most $\epsilon$ on the target $\bra{\phi} A^{-1} \ket{\psi}$. The number of controlled Pauli rotations per sample is $N_{\rm CP} = 2 r L$, as in \cref{thm:trotter_sample_complexity}.
\end{corollary}
\begin{proof}
    The proof is given in Supplementary Information~\ref{app:cor_trotter_rmse}
\end{proof}

We have stated \cref{cor:trotter_rmse} with the bias condition $|B_{\rm PF}| < \epsilon/2$ to match the hypothesis of \cref{thm:trotter_sample_complexity}, but this is one of many admissible choices. The corollary controls the MSE budget $(\epsilon/2)^2$, which can be allocated freely between the bias-squared term $|B_{\rm PF}|^2$ and the variance term $2 N_y^2 N_z^2/(\lambda^2 N_{\rm S})$, and any allocation summing to $(\epsilon/2)^2$ yields a valid set of conditions on $r$ and $N_{\rm S}$.

Different allocations correspond to different points along the trade-off between $N_{\rm S}$ and $N_{\rm CP}$. From \cref{thm:trotter_bias_bound}, $|B_{\rm PF}| \propto r^{-2}$, so allocating less of the budget to the bias requires a larger $r$ and thus a larger $N_{\rm CP} = 2rL$, while simultaneously freeing up budget for the variance and reducing $N_{\rm S}$. Minimizing the total complexity $N_{\rm S} \cdot N_{\rm CP}$ over this trade-off identifies an optimal allocation of $|B_{\rm PF}|^2 = (\epsilon/2)^2/5$ for the bias and $4(\epsilon/2)^2/5$ for the variance.

\subsection{Error bounds for the random Taylor expansion subroutine}
\label{sec:resource_req_rte}

We now describe the second method used to perform controlled time evolution, i.e. the \RTE\ approach. In particular, we consider the sampling based approach presented in \cite{Wan:2021non} to construct an estimator for the exponential $e^{-\imi \tilde{A}\tau}$. We use a Taylor series expansion to approximately represent the exponential as a Linear Combination of Unitaries (LCU) and sample individual unitaries according to the coefficients of the LCU. The controlled version of each such unitary can then be used in the Hadamard test circuit step shown in \cref{fig:alg_outline} to estimate the overlap $\braket{\phi|\tilde A^{-1}|\psi}$. This approach was also used in~\cite{Wang:2023els} for randomized quantum linear algebra algorithms.

Similar to the \PF\ method (see~\cref{eq:pf_split}), we divide the Fourier time $\tau$ into $r$ segments and write
\begin{align}
    e^{-\imi \tilde{A}\tau} = \left(e^{-\imi \tilde{A}\tau/r}\right)^r.
\end{align}
As shown in Supplementary Information~\ref{app_rte} and~\cite{Wan:2021non}, we can Taylor expand each exponential $e^{-\imi \tilde{A}\tau/r}$, pair consecutive terms in the expansion to absorb the imaginary unit into a Pauli rotation, and truncate the resulting infinite series to obtain the approximation
\begin{align}
    \label{eq:rte_finite}
    e^{-\imi \tilde{A}\tau/r}
    & \approx \sum_{\substack{n=0\\n \text{ even}}}^{\nmax} \frac{1}{n!} \left(\frac{\imi\tau}{r}\right)^n \sqrt{1 + \left(\frac{\tau / r}{n+1}\right)^2} \left(\sum_{\ell \in [L]} \frac{c_{\ell}}{\lambda}P_{\ell}\right)^{n}\sum_{\ell \in [L]} \frac{c_{\ell}}{\lambda}\, e^{-\imi \theta_n P_{\ell}},
\end{align}
where $\nmax \in \{0, 2, 4, \dots\}$ is an even truncation parameter and $\theta_n$ is given by
\begin{align}
    \theta_n =  \cos^{-1}\!\left(\left(1 + \left(\frac{\tau / r}{n+1}\right)^2\right)^{-1/2}\right).
\end{align}
Since $n$ is even, the factor $(\imi\tau/r)^n$ is real, so that each summand in \cref{eq:rte_finite} has the structure of a (signed) scalar times a product of Paulis times a single-qubit Pauli rotation $e^{-\imi\theta_n P_\ell}$.

We can therefore rewrite \cref{eq:rte_finite} as a finite linear combination of unitaries (LCU),
\begin{align}
    \label{eq:rte_r}
    e^{-\imi \tilde{A}\tau/r} \approx \sum_{m \in \mathcal{M}'}d_{m\tau}'U_{m\tau}',
\end{align}
where $\mathcal{M}'$ is an appropriately defined index set and each coefficient $d_{m\tau}' \in \mathbb{R}^+$ is obtained by absorbing any sign or phase into the corresponding unitary $U_{m\tau}' \sim P_{\ell_1}P_{\ell_2}\cdots P_{\ell_n}\, e^{-\imi \theta_n P_{\ell_{n+1}}}$, with $\ell_1, \ell_2, \dots, \ell_{n+1} \in [L]$ and $n \in \{0, 2, \dots, \nmax\}$ (so for $n=0$ the leading product of Paulis is empty and only the rotation $e^{-\imi\theta_0 P_{\ell_1}}$ remains). Each unitary $U_{m\tau}'$ thus consists of at most $\nmax$ Clifford operations together with a single multi-qubit Pauli rotation, the latter of which can be efficiently synthesized into a sequence of Clifford operations (CNOTs) and one single-qubit Pauli rotation via the standard CNOT ladder construction~\cite{Nielsen_Chuang_2010}. It follows that the non-Clifford cost of a controlled-$U_{m\tau}'$ is equivalent to that of controlling one single-qubit Pauli rotation.

The LCU approximation for $e^{-\imi\tilde{A}\tau}$ is then obtained by raising \cref{eq:rte_r} to the $r^{\text{th}}$ power,
\begin{equation}
    \label{eq:rte_lcu}
    e^{-\imi\tilde{A}\tau} \approx \left(\sum_{m \in \mathcal{M}'}d_{m\tau}'U_{m\tau}'\right)^r =: \sum_{m\in\mathcal{M}}d_{m\tau}U_{m\tau},
\end{equation}
for the appropriately defined index set $\mathcal{M} \cong \left(\mathcal{M}'\right)^r$. A controlled-$U_{m\tau}$ thus has non-Clifford cost equivalent to controlling $r$ single-qubit Pauli rotations.

We now consider the problem of sampling individual unitaries of the form $U_{m\tau}$ from the LCU in \cref{eq:rte_lcu}. We begin by rewriting \cref{eq:rte_r} as
\begin{align}
    \label{eq:lcu_prob}
    e^{-\imi \tilde{A}\tau/r} &\approx \alpha(\tau) \sum_{m \in \mathcal{M}'}\frac{d_{m\tau}'}{\alpha(\tau)}U_{m\tau}'
    =: \alpha(\tau) \sum_{m \in \mathcal{M}'}p_{m\tau}'U_{m\tau}',
\end{align}
where $\alpha(\tau) := \sum_{m \in \mathcal{M}'}d_{m\tau}'$ is the one-norm of the per-segment LCU coefficients and $\{p_{m\tau}'\}_{m\in\mathcal{M}'}$ defines a probability distribution over $\mathcal{M}'$. To sample a unitary $U_{m\tau}$ from \cref{eq:rte_lcu} such that $\mathbb{E}[U_{m\tau}] = e^{-\imi \tilde{A}\tau}$, we draw $r$ unitaries $U_{m\tau}'$ independently according to the probabilities $p_{m\tau}'$ and form their product. The overall sampling normalization constant for this procedure, denoted $\Nrte$, is
\begin{align}
    \label{eq:nrte_ub_0}
    \Nrte = \alpha(\tau)^r \leq e^{\tau^2/r},
\end{align}
where the upper bound is established in Lemma~2 of~\cite{Wan:2021non}. Algorithm~2 of~\cite{Wan:2021non} provides an efficient classical procedure for sampling and returning a description of $r$ such unitaries $U_{m\tau}'$ according to the probabilities $\{p_{m\tau}'\}$ without having to construct the full set of probabilities $\{p_{m\tau}'\}_{m \in \mathcal{M}'}$ explicitly.

Up to this point, all computations are performed classically. As in the product formula case, we construct an estimator for $\bra{\phi}U_{m\tau}\ket{\psi}$ using the Hadamard test circuit applied to $U_{\phi}^{\dagger}U_{m\tau}U_{\psi}$. We perform the test twice---once for the real part and once for the imaginary part---yielding single-shot outcomes $\hat X_{\tau}, \hat Y_{\tau} \in \{-1, +1\}$ that are unbiased estimators of these parts conditional on $\tau$ and the sampled unitary $U_{m\tau}$:
\begin{align}
    \label{eq:rte_estimators}
    \mathbb{E}[\hat X_{\tau} \mid \tau, U_{m\tau}] &= \Re \bra{0}U_{\phi}^{\dagger}U_{m\tau}U_{\psi}\ket{0}, &
    \mathbb{E}[\hat Y_{\tau} \mid \tau, U_{m\tau}] &= \Im \bra{0}U_{\phi}^{\dagger}U_{m\tau}U_{\psi}\ket{0}.
\end{align}
Combining these with the phase factor $\omega_{\tau}$, the sampling normalization $\Nrte$ from \cref{eq:nrte_ub_0}, and the prefactor $N_y N_z / \lambda$ from \cref{eq:fourier_series_}, we form the single-sample estimator
\begin{align}
    \label{eq:biased_estimator}
    \hat Z := \omega_{\tau} \frac{\Nrte\, N_y N_z}{\lambda} \left(\hat X_{\tau} + \imi\, \hat Y_{\tau}\right),
\end{align}
where the randomness in $\hat Z$ comprises the sampled Fourier time $\tau$, the sampled unitary $U_{m\tau}$, and the two Hadamard outcomes; the expectation $\mathbb{E}[\hat Z]$ is taken jointly over all three. We define the bias
\begin{align}
    \label{eq:bias}
    B_{\rm RTE} := \frac{\bra{\phi} I_D^{(\epsilon_D, \epsilon_T)}(\tilde A) \ket{\psi}}{\lambda} - \mathbb{E}[\hat Z],
\end{align}
which captures the systematic error introduced by truncating the Taylor expansion of the time evolution operator at order $\nmax$ in \cref{eq:rte_finite}. We bound this bias as follows.

\begin{theorem}
    \label{thm:bias_bound}
    For any $r \geq t_{\max}$ and any even truncation parameter $\nmax \in \{0, 2, 4, \dots\}$, the bias defined in \cref{eq:bias} satisfies
    \begin{align}
    \label{eq:simple_rc_bias}
    |B_{\rm RTE}| \leq \frac{\sqrt{2}\, r\, N_y N_z}{\lambda}\, e^{t_{\max}^2/r} \left(\frac{e\, t_{\max}}{r\, (\nmax+1)}\right)^{\nmax+1},
    \end{align}
    where $N_y$ and $N_z$ are given by \cref{lemma:normalization_bounds}, and $t_{\max}$ is given by \cref{eq:max_fourier_time}.
\end{theorem}
\begin{proof}
    The proof is given in Supplementary Information~\ref{app:thm_bias_bound}.
\end{proof}
We observe that for fixed $\tau, r, N_y, N_z$, and $\lambda$ the upper bound on $\Brte$ falls off super exponentially as a function of $\nmax$. We now present efficiently computable choices for the sample and non-Clifford gate complexities required by the \RTE\ approach given some error budget:
\begin{theorem}
    \label{thm:random_sample_complexity}
    Let $Z_{N_{\rm S}}$ be the random variable obtained by averaging $N_{\rm S}$ independent samples of the biased estimator $\hat Z$ defined in \cref{eq:biased_estimator}, and assume the bias defined in \cref{eq:bias} satisfies $|B_{\rm RTE}| < \epsilon/2$ (\cref{thm:bias_bound}). If $\epsilon_D + \epsilon_T \leq \epsilon/2$, then $Z_{N_{\rm S}}$ solves \cref{problem_statement} provided
    \begin{align}
        \label{eq:rc_complexities}
        N_{\rm S} = 1 + \ln\!\left(\frac{4}{\delta}\right) \frac{16\, e^{2 t_{\max}^2/r}\, N_y^2 N_z^2}{\lambda^2 (\epsilon/2 - |B_{\rm RTE}|)^2}, \qquad
        N_{\rm CP} = r,
    \end{align}
    where $N_{\rm CP}$ provides an upper bound on the number of controlled Pauli rotations required per sample and $N_y, N_z$ are given by \cref{lemma:normalization_bounds}.
\end{theorem}
\begin{proof}
The proof is given in Supplementary Information~\ref{app:random_sample_complexity}.
\end{proof}

If we choose $r \geq t_{\max}$, the condition $|B_{\rm RTE}| < \epsilon/2$ can be satisfied by setting the upper bound from \cref{eq:simple_rc_bias} below $\epsilon/2$, namely
\begin{align}
    \label{eq:rte_bias_condition}
    \frac{\sqrt{2}\, r\, N_y N_z}{\lambda}\, e^{t_{\max}^2 / r}\!\left(\frac{e\, t_{\max}}{r\, (\nmax+1)}\right)^{\nmax+1} < \frac{\epsilon}{2},
\end{align}
from which an appropriate value of $\nmax$ can be determined.

A salient feature of \cref{thm:random_sample_complexity} is that $N_{\rm S}$ scales exponentially in the ratio $t_{\max}^2 / r$. Choosing $r = t_{\max}$ thus implies the need for exponentially many samples in $t_{\max}$ while keeping the per-sample gate depth constant. Choosing $r \propto t_{\max}^2$ instead, together with sufficiently large $\nmax$ via \cref{thm:bias_bound}, essentially strips $N_{\rm S}$ of its dependence on $t_{\max}^2$ at the cost of giving $N_{\rm CP}$ a quadratic dependence on $t_{\max}$. There is therefore an explicit trade-off between the sample complexity and the per-sample gate complexity which can be navigated by selecting $r$ adaptively based on each sampled Fourier time rather than fixing it as a function of $t_{\max}$. As noted in Appendix~D of~\cite{Wan:2021non}, one can in principle compute the elements $\{r_{j,k}\}_{(j,k) \in [J] \times [K]}$, assigning a unique value of $r$ to each Fourier time, that minimize either the total complexity $N_{\rm S}\, N_{\rm CP}$ or one of $N_{\rm S}, N_{\rm CP}$ subject to a constraint on the other, by solving a one-dimensional optimization problem. Conversely, choosing $r = \mathcal{O}(1)$ independent of $t_{\max}$ renders the
task of generating $N_{\rm S}$ samples computationally intractable for many
problems of practical interest, due to the exponential dependence on $t_{\max}^2$;
this exponential is the only one appearing anywhere in our analysis, and the
choice $r \propto t_{\max}^2$ suppresses it to a constant, after
which all costs are polynomial in $\kappa$.

\Cref{thm:random_sample_complexity} guarantees $|Z_{N_{\rm S}} - \lambda^{-1}\bra{\phi} I_D^{(\epsilon_D,\epsilon_T)}(\tilde A) \ket{\psi}| \leq \epsilon/2$ with probability at least $1 - \delta$, at the cost of a logarithmic overhead in $\delta$ inherited from Hoeffding's inequality. The following corollary translates \cref{thm:random_sample_complexity} into an RMSE-based statement, paralleling \cref{cor:trotter_rmse} for the PF subroutine.

\begin{corollary}
\label{cor:rte_rmse}
Let $Z_{N_{\rm S}}$ be the average of $N_{\rm S}$ independent samples of the biased estimator $\hat Z$ defined in \cref{eq:biased_estimator}, and let
\begin{equation}
    \mathrm{RMSE}(Z_{N_{\rm S}}) := \sqrt{\,\mathbb{E}\!\left[\,\left| Z_{N_{\rm S}} - \frac{\bra{\phi} I_D^{(\epsilon_D,\epsilon_T)}(\tilde A) \ket{\psi}}{\lambda} \right|^2\,\right]\,}.
\end{equation}
Then
\begin{equation}
    \mathrm{RMSE}(Z_{N_{\rm S}}) \;\leq\; \sqrt{\,|B_{\rm RTE}|^2 + \frac{2\, e^{2 t_{\max}^2/r}\, N_y^2 N_z^2}{\lambda^2\, N_{\rm S}}\,}.
\end{equation}
In particular, assuming the bias satisfies $|B_{\rm RTE}| < \epsilon/2$ (\cref{thm:bias_bound}), the choice
\begin{equation}
    N_{\rm S} \;\geq\; \frac{8\, e^{2 t_{\max}^2/r}\, N_y^2 N_z^2}{\lambda^2\, \big((\epsilon/2)^2 - |B_{\rm RTE}|^2\big)}
\end{equation}
ensures $\mathrm{RMSE}(Z_{N_{\rm S}}) \leq \epsilon/2$. Combined with $\epsilon_D + \epsilon_T \leq \epsilon/2$, this guarantees an RMSE of at most $\epsilon$ on the target $\bra{\phi} A^{-1} \ket{\psi}$. The number of controlled Pauli rotations per sample is $N_{\rm CP} = r$, as in \cref{thm:random_sample_complexity}.
\end{corollary}
\begin{proof}
The proof is given in Supplementary Information~\ref{app:cor_rte_rmse}.
\end{proof}

As with \cref{cor:trotter_rmse}, we have stated \cref{cor:rte_rmse} with the bias condition $|B_{\rm RTE}| < \epsilon/2$ to match the hypothesis of \cref{thm:random_sample_complexity}, but the MSE budget $(\epsilon/2)^2$ can be allocated freely between the bias-squared term $|B_{\rm RTE}|^2$ and the variance term $2 e^{2 t_{\max}^2/r} N_y^2 N_z^2/(\lambda^2 N_{\rm S})$.

Unlike the PF case, however, allocating less of the budget to the bias does not increase the gate complexity per sample: \cref{thm:bias_bound} shows that $|B_{\rm RTE}|$ can be made arbitrarily small by increasing the truncation parameter $\nmax$, which does not enter $N_{\rm CP}$ but only the classical sampling cost of \cref{eq:rte_lcu}. The optimal allocation for $N_{\rm CP}$-aware total complexity is therefore $|B_{\rm RTE}| \to 0$, with the entire MSE budget devoted to the variance term.

\subsection{Numerical results}
\label{sec:numerics}

We first compare the theoretical bounds on the Fourier series approximation derived in \cref{thm:truncation_parameters,thm:fourier_parameters} with empirical results for randomly generated Hermitian matrices of prescribed condition number. We construct each test matrix as $A := UDU^{\dagger}$, where $D$ is a diagonal matrix with two entries pinned to $1$ and $1/\kappa$ and any remaining entries drawn uniformly from $[-1, -1/\kappa] \cup [1/\kappa, 1]$. The pinned entries ensure $\lVert A\rVert = 1$ and $\lVert A^{-1}\rVert = \kappa$ exactly, so that $\kappa(A) = \kappa$, while the remaining entries make the construction non-trivial. The unitary $U$ is sampled from the Haar measure on $U(N)$. Since the bounds in \cref{thm:truncation_parameters,thm:fourier_parameters} have no explicit dependence on the dimension $N$, we test their dimensional independence by repeating the experiment at two different sizes, $N = 4$ and $N = 16$. We generate $100$ random matrices at each size for each pair
$$(\kappa, \epsilon_F) \in \lbrace 10, 10^2, 10^3 \rbrace \times \lbrace 10^{-2}, 10^{-3}, 10^{-4} \rbrace,$$
and for each matrix compute the Pauli weight $\lambda$ and the rescaled condition number $\tilde{\kappa}^* = \lambda \kappa$. Across the $900$ trials at $N = 4$ we find $\lambda \in [1.43, 3.08]$ and hence $\tilde\kappa^* \in [14.3, 3.1 \times 10^3]$; across the $900$ trials at $N = 16$ we find $\lambda \in [4.89, 9.91]$ and hence $\tilde\kappa^* \in [54, 9.3 \times 10^3]$.

We set $\epsilon_T = \epsilon_D = \epsilon_F/2$ and apply \cref{alg:fourier_series_construction} to construct the Fourier series approximation
\begin{align}
    A^{-1} \approx F(A) := \lambda^{-1} I_D^{(\epsilon_F/2, \epsilon_F/2)}(\tilde{A}),
\end{align}
where $\tilde A = \lambda^{-1} A$. By \cref{alg:fourier_series_construction}, the spectral-norm error satisfies
\begin{align}
    \label{eq:numerics_fourier_error}
    \lambda \cdot \big\|A^{-1} - \lambda^{-1} I_D^{(\epsilon_F/2,\epsilon_F/2)}(\tilde{A})\big\| \leq \epsilon_F.
\end{align}
For each trial we evaluate the left-hand side of \cref{eq:numerics_fourier_error} directly. Inequality \cref{eq:numerics_fourier_error} is satisfied for every one of the $1800$ trial matrices across both dimensions. \cref{table:num_fourier_error_bounds_4x4,table:num_fourier_error_bounds_16x16} report the Fourier parameters $(J, K)$ and the minimum, mean, and maximum empirical errors per cell at $N = 4$ and $N = 16$, respectively. The reported $(J, K)$ are those produced by \cref{thm:fourier_parameters} at the largest $\tilde\kappa^*$ realized in the cell, so that the same $(J, K)$ suffices for every matrix in the cell.

Three observations follow. First, the empirical errors are tightly concentrated within each cell: the spread between the minimum and maximum across $100$ trials never exceeds $11\%$ of the mean at either dimension, indicating that the worst-case error in each cell is statistically meaningful and not driven by outlier matrices. Second, the empirical maximum is consistently between $17\%$ and $21\%$ of the bound $\epsilon_F$ across all cells and both dimensions, with average $19\%$, demonstrating that the bounds in \cref{thm:truncation_parameters,thm:fourier_parameters} are valid. Third, the ratio $\epsilon_{\max}/\epsilon_F$ shows no systematic trend with $N$: the bound is satisfied at the same fractional margin at $N = 16$ as at $N = 4$, confirming the dimensional independence predicted by the theory. As expected, however, the number of Fourier terms $JK$ required to reach a given $\epsilon_F$ at fixed $\kappa$ is larger at $N = 16$ than at $N = 4$, reflecting the larger typical $\tilde\kappa^* = \lambda \kappa$ at the higher dimension.

\begin{table}[h!tbp]
    \centering
    \begin{tabular}{>{\columncolor{gray!20}}cc|ccccc}
        \toprule
        \rowcolor{gray!40} \textbf{$\kappa$} & \textbf{$\epsilon_F$} & \textbf{$J$} & \textbf{$K$} & \textbf{$\epsilon_{\min}$} & \textbf{$\epsilon_{\rm mean}$} & \textbf{$\epsilon_{\max}$} \\
        \midrule
            10 & \multirow{3}{*}{$10^{-2}$} &   158 &   175 & $1.73 \times 10^{-3}$ & $1.80 \times 10^{-3}$ & $1.90 \times 10^{-3}$ \\
        $10^2$ &                            &  1703 &  2100 & $1.75 \times 10^{-3}$ & $1.82 \times 10^{-3}$ & $1.93 \times 10^{-3}$ \\
        $10^3$ &                            & 20997 & 26177 & $1.97 \times 10^{-3}$ & $2.06 \times 10^{-3}$ & $2.09 \times 10^{-3}$ \\
        \midrule
            10 & \multirow{3}{*}{$10^{-3}$} &   212 &   239 & $1.73 \times 10^{-4}$ & $1.79 \times 10^{-4}$ & $1.86 \times 10^{-4}$ \\
        $10^2$ &                            &  2068 &  2552 & $1.97 \times 10^{-4}$ & $2.05 \times 10^{-4}$ & $2.08 \times 10^{-4}$ \\
        $10^3$ &                            & 26338 & 32840 & $1.86 \times 10^{-4}$ & $1.95 \times 10^{-4}$ & $2.04 \times 10^{-4}$ \\
        \midrule
            10 & \multirow{3}{*}{$10^{-4}$} &   239 &   270 & $1.92 \times 10^{-5}$ & $1.99 \times 10^{-5}$ & $2.01 \times 10^{-5}$ \\
        $10^2$ &                            &  2319 &  2862 & $1.85 \times 10^{-5}$ & $1.94 \times 10^{-5}$ & $2.02 \times 10^{-5}$ \\
        $10^3$ &                            & 24983 & 31146 & $1.72 \times 10^{-5}$ & $1.73 \times 10^{-5}$ & $1.77 \times 10^{-5}$ \\
        \bottomrule
    \end{tabular}
    \caption{Numerical verification of the bounds on the Fourier series approximation to the matrix inverse derived in \cref{thm:truncation_parameters,thm:fourier_parameters} for $4 \times 4$ Hermitian matrices with condition number $\kappa$ and approximation tolerance $\epsilon_F$. For each cell, $100$ random matrices are generated and approximated using \cref{alg:fourier_series_construction} with $\epsilon_T = \epsilon_D = \epsilon_F/2$. The table reports the Fourier parameters $(J, K)$ corresponding to the largest rescaled condition number $\tilde\kappa^* = \lambda \kappa$ realized in the cell, together with the minimum, mean, and maximum of the empirical spectral-norm error $\lambda \|A^{-1} - \lambda^{-1} I_D^{(\epsilon_F/2,\epsilon_F/2)}(\tilde A)\|$ across the $100$ trials.}
    \label{table:num_fourier_error_bounds_4x4}
\end{table}

\begin{table}[h!tbp]
    \centering
    \begin{tabular}{>{\columncolor{gray!20}}cc|ccccc}
        \toprule
        \rowcolor{gray!40} \textbf{$\kappa$} & \textbf{$\epsilon_F$} & \textbf{$J$} & \textbf{$K$} & \textbf{$\epsilon_{\min}$} & \textbf{$\epsilon_{\rm mean}$} & \textbf{$\epsilon_{\max}$} \\
        \midrule
            10 & \multirow{3}{*}{$10^{-2}$} &   559 &    674 & $2.08 \times 10^{-3}$ & $2.11 \times 10^{-3}$ & $2.12 \times 10^{-3}$ \\
        $10^2$ &                            &  6068 &   7546 & $1.83 \times 10^{-3}$ & $1.93 \times 10^{-3}$ & $2.00 \times 10^{-3}$ \\
        $10^3$ &                            & 72979 &  91053 & $1.73 \times 10^{-3}$ & $1.73 \times 10^{-3}$ & $1.78 \times 10^{-3}$ \\
        \midrule
            10 & \multirow{3}{*}{$10^{-3}$} &   647 &    781 & $1.87 \times 10^{-4}$ & $1.94 \times 10^{-4}$ & $2.00 \times 10^{-4}$ \\
        $10^2$ &                            &  7423 &   9234 & $1.73 \times 10^{-4}$ & $1.73 \times 10^{-4}$ & $1.79 \times 10^{-4}$ \\
        $10^3$ &                            & 84522 & 105458 & $1.83 \times 10^{-4}$ & $1.91 \times 10^{-4}$ & $2.01 \times 10^{-4}$ \\
        \midrule
            10 & \multirow{3}{*}{$10^{-4}$} &   821 &    994 & $1.73 \times 10^{-5}$ & $1.73 \times 10^{-5}$ & $1.75 \times 10^{-5}$ \\
        $10^2$ &                            &  8388 &  10435 & $1.84 \times 10^{-5}$ & $1.90 \times 10^{-5}$ & $2.01 \times 10^{-5}$ \\
        $10^3$ &                            & 93124 & 116190 & $1.95 \times 10^{-5}$ & $2.01 \times 10^{-5}$ & $2.05 \times 10^{-5}$ \\
        \bottomrule
    \end{tabular}
    \caption{Numerical verification of the bounds on the Fourier series approximation to the matrix inverse derived in \cref{thm:truncation_parameters,thm:fourier_parameters} for $16 \times 16$ Hermitian matrices, with all parameters and conventions as in \cref{table:num_fourier_error_bounds_4x4}. The Fourier parameters $(J, K)$ are larger than the corresponding entries at $N = 4$ due to the larger Pauli weight typical of higher-dimensional matrices, but the empirical errors satisfy the same bound $\epsilon_F$ at the same fractional margin.}
    \label{table:num_fourier_error_bounds_16x16}
\end{table}

We now turn to the numerical verification of the RMSE bound of
\cref{cor:trotter_rmse} for the second-order PF estimator $\hat Z$ defined
in \cref{eq:PF_estimator}, on small-scale problem instances drawn from the
same family used in the verification of the Fourier series error bounds earlier in this section. We restrict to
$N = 4$, fix $\kappa = 100$ and $\epsilon_F = 2 \times 10^{-3}$, and
generate $10$ independent random matrices. We take
$\ket{\psi} = \ket{\phi} = \ket{0}$, so the target scalar is the
diagonal matrix element $\langle 0 | A^{-1} | 0 \rangle$ and the
state-preparation unitaries $U_\psi, U_\phi$ in \cref{eq:PF_estimator}
reduce to identity, isolating the empirical convergence behaviour from
the cost of state preparation. For each instance we set
$\epsilon_T = \epsilon_D = \epsilon_F/2$ and apply
\cref{alg:fourier_series_construction} to obtain the Fourier parameters
$(N_y, N_z, J, K)$ and the joint sampling distribution $p_y \otimes p_z$
over Fourier nodes from \cref{eq:probs}. Across the $10$ instances we find
$\lambda \in [1.52,\, 2.77]$, hence $\tilde\kappa^* = \lambda \kappa \in
[152,\, 277]$, and $t_{\max} \in [3970,\, 7540]$.

In our numerical simulations we make the practical choice of selecting the
Trotter number adaptively per sample, taking
\begin{align}
    \label{eq:adaptive_pf_rule}
    r(\tau) := \max\!\left(1,\, \left\lceil c\, |\tau|^{3/2} \right\rceil\right),
    \qquad
    c := \sqrt{\frac{2 N_y N_z f}{\lambda \epsilon_F}},
\end{align}
where $f$ is the second-order Trotter constant from \cref{lemma:trotter_number}.
This rule sets $r(t_{\max})$ equal to the constant lower bound in
\cref{eq:trotter_number_lower_bound} and uniformizes the per-sample Trotter
error, $f \tau^3 / r(\tau)^2 \leq f / c^2$, across the Fourier distribution.
The bias bound from \cref{thm:trotter_bias_bound} is therefore satisfied at
the same level $|B_{\rm PF}| \leq \epsilon_F/2$ as under the constant choice,
with smaller average gate count per sample.

For each instance we precompute the $J \times K$ grid of
expectation values $\langle 0 | S_{r(t_{jk})}(t_{jk}) | 0 \rangle$, then
simulate $N_S$ samples of $\hat Z$ by drawing $(j, k)$ from
$p_y \otimes p_z$ and adding independent unit-variance Gaussian noise to
each grid value as a conservative model of the single-shot Hadamard
variance, which is at most unity per shot. We sweep $N_S$ on
a logarithmic grid from $10^2$ to $10^8$ and compute the empirical RMSE of
$\Re(\hat Z_{N_S})$ against the target $\langle 0 | A^{-1} | 0 \rangle$
from $20$ independent Monte Carlo runs at each $N_S$.

The dashed curve in \cref{fig:pf_estimator_convergence} reports the
analytical upper bound on the RMSE of $\hat Z_{N_S}$ relative to the true
target $\langle 0|A^{-1}|0\rangle$ that the analysis of
\cref{cor:trotter_rmse} guarantees. We combine the worst-case Trotter bias
allocation $|B_{\rm PF}| \leq \epsilon_F/2$ from
\cref{thm:trotter_bias_bound}, the closed-form variance bound from \cref{cor:trotter_rmse},
and the Fourier approximation error $\|F(A) - A^{-1}\| \leq \epsilon_F$ from
\cref{alg:fourier_series_construction} via the triangle inequality, to compute the per-instance bound. We plot the maximum of this expression over the $10$
instances so that a single curve uniformly upper-bounds all $10$ empirical
trajectories. In particular, we plot
\begin{align}
    \label{eq:max_analytical_bound_pf}
    \mathrm{Max}\hspace{1mm}\mathrm{analytical}\hspace{1mm}\mathrm{bound}\hspace{1mm}\mathrm{across}\hspace{1mm}\mathrm{instances}\hspace{1mm} = \max_{i\in[10]}\left(\sqrt{(\epsilon_F/2)^2 + 2 (N_y^{(i)} N_z^{(i)})^2 / ((\lambda^{(i)})^2 N_S)}\right)
+ \epsilon_F.
\end{align}

\begin{figure}[h]
    \centering
    \includegraphics[width=0.8\linewidth]{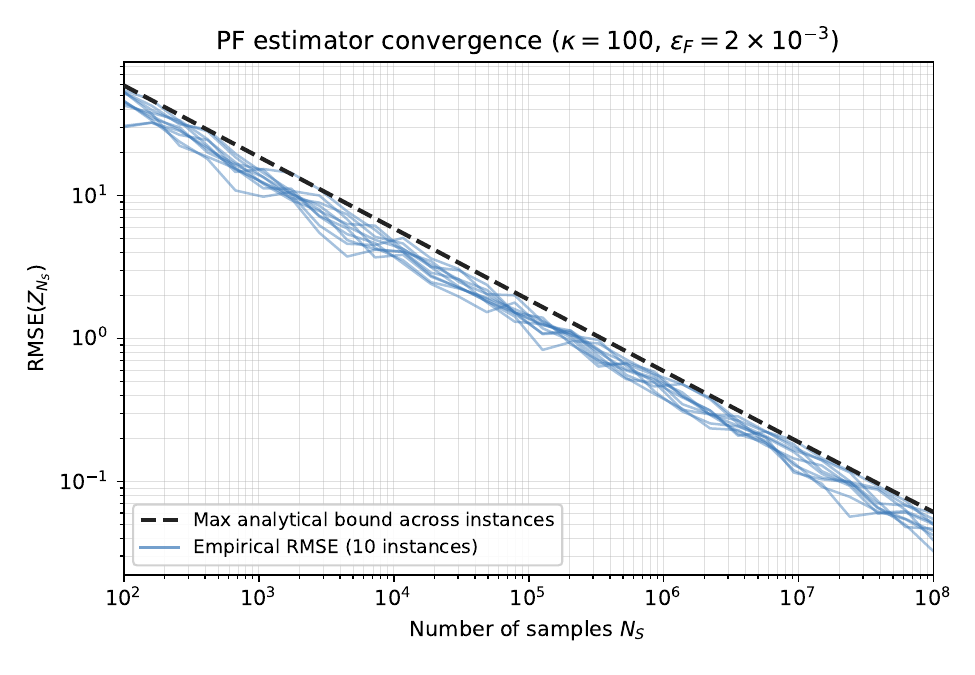}
    \caption{Empirical convergence of the PF estimator $\hat Z$ defined in
    \cref{eq:PF_estimator}, on $10$ random Hermitian matrices of dimension
    $N = 4$ with condition number $\kappa = 100$ and Fourier error
    $\epsilon_F = 2 \times 10^{-3}$, using the adaptive Trotter rule
    \cref{eq:adaptive_pf_rule}. Each thin curve shows the empirical RMSE
    of $\Re(\hat Z_{N_S})$ against the target
    $\langle 0 | A^{-1} | 0 \rangle$, computed from $20$ independent Monte
    Carlo runs at each $N_S$. The dashed curve is the maximum across the
    $10$ instances of the analytical RMSE upper bound combining
    \cref{cor:trotter_rmse} and the Fourier approximation error of
    \cref{alg:fourier_series_construction}, and provides a uniform upper
    bound on the empirical RMSE relative to the true target
    $\langle 0|A^{-1}|0\rangle$.}
    \label{fig:pf_estimator_convergence}
\end{figure}

\Cref{fig:pf_estimator_convergence} reports the $10$ empirical RMSE curves
together with the analytical RMSE bound described above. Across all $10$
instances and all $N_S$ on the grid, the empirical RMSE tracks the
analytical bound: the curves descend at the variance-controlled rate
$\propto 1/\sqrt{N_S}$, cluster tightly across instances, and lie below
the bound at the overwhelming majority of datapoints, with less than $1\%$ of the measured points exceeding the bound. These rare crossings are consistent with the Monte-Carlo
fluctuation of the empirical RMSE estimator itself, which is computed
from only $20$ independent runs at each $N_S$, along with our conservative choice for the variance of the single-shot Hadamard test. The mean empirical
RMSE-to-bound ratio of $0.75$ across all datapoints reflects three
sources of slack in the analytical bound: the worst-case Trotter bias
allocation $|B_{\rm PF}| \leq \epsilon_F/2$ over-bounds the actual
deterministic bias of the Fourier-summed estimator; the input Fourier
error $\epsilon_F$ over-bounds the actual approximation error
$\|F(A) - A^{-1}\|$ (as shown earlier in this section); and the triangle inequality
that combines the deterministic-bias and stochastic-variance
contributions is loose when the two contributions partially cancel
rather than add. \Cref{cor:trotter_rmse} is therefore empirically
validated for the PF subroutine, with the analytical bound holding
uniformly across the $10$ instances and remaining within a factor of
order unity of the empirical RMSE throughout the $N_S$ sweep.

We now turn to the Fourier-time distribution shown in \cref{fig:fourier_time_distribution}: across the $10$ instances,
the median sampled time satisfies $|\tau| \in [384,\, 715]$, the
$99$th-percentile time satisfies $|\tau| \in [1.8 \times 10^3,\,
3.3 \times 10^3]$, and $t_{\max}$ ranges up to $\sim 7.5 \times 10^3$.
This distribution governs the per-sample gate cost of \emph{both}
estimators on a quantum device, since each draws a Fourier time from
$p_y \otimes p_z$ and constructs a controlled time evolution at that time.

\begin{figure}[h]
    \centering
    \includegraphics[width=0.8\linewidth]{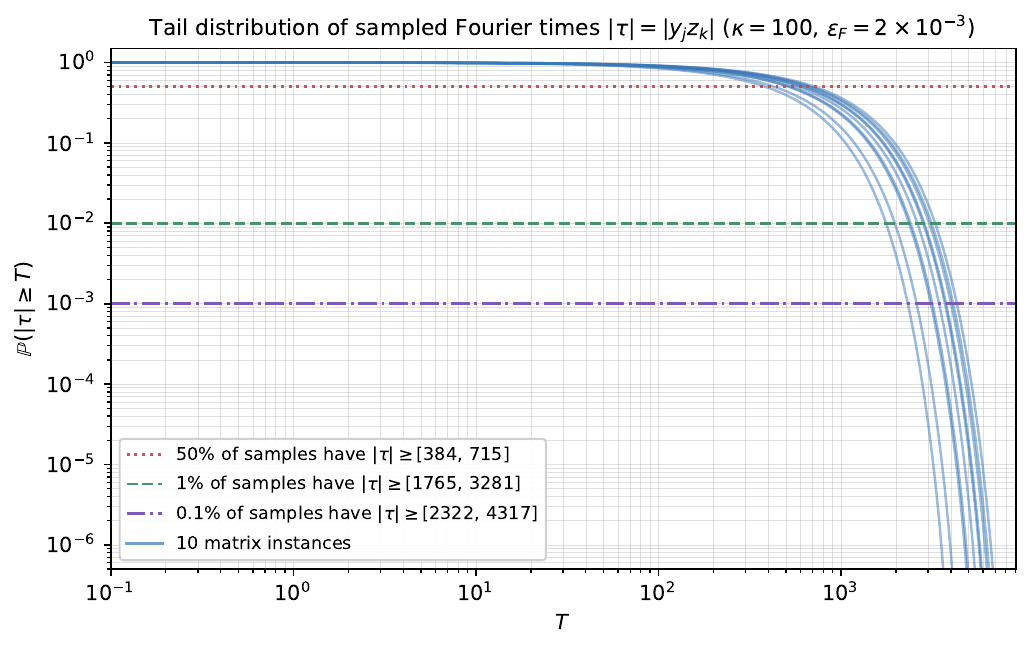}
    \caption{Tail distribution of the absolute value of the sampled Fourier times
    $|\tau|$ under the joint sampling distribution
    $p_y \otimes p_z$ from \cref{eq:probs}, for the same $10$ random
    Hermitian matrices as in \cref{fig:pf_estimator_convergence}. Each thin
    curve gives the complementary cumulative distribution function
    $\mathbb{P}(|\tau| \geq T)$ for one matrix instance. The horizontal
    lines mark the corresponding ranges of the $50\%$, $1\%$, and $0.1\%$
    quantiles of $|\tau|$ across the $10$ instances.}
    \label{fig:fourier_time_distribution}
\end{figure}

For the PF subroutine, the adaptive rule \cref{eq:adaptive_pf_rule} sets
the per-sample gate count to $N_{\rm CP}^{\rm PF}(\tau) = 2 L\, r(\tau)$
controlled-Pauli rotations. Combining this with the distribution in
\cref{fig:fourier_time_distribution} yields per-sample expected gate
counts $\mathbb{E}_{\tau}[N_{\rm CP}^{\rm PF}(\tau)] \in [3.4 \times 10^7,\,
1.1 \times 10^8]$ across the $10$ instances, with the worst-case sample
(at $\tau = t_{\max}$) requiring up to $\sim 2 \times 10^9$ rotations.
For the RTE subroutine, controlling the LCU sampling normalization
$\Nrte(\tau) = \alpha(\tau)^r \leq e^{\tau^2/r}$ from \cref{eq:nrte_ub_0}
--- which enters the per-sample variance as $\Nrte(\tau)^2$ ---
uniformly across the Fourier distribution requires the per-sample
Trotter count to satisfy $r_{\rm RTE}(\tau) \gtrsim \tau^2$. With the
choice $r_{\rm RTE}(\tau) = \lceil \tau^2 \rceil$, the per-sample gate
count $N_{\rm CP}^{\rm RTE}(\tau) = r_{\rm RTE}(\tau)$ exceeds $1.5 \times
10^5$ for more than half of all samples, exceeds $3 \times 10^6$ for
$1\%$ of samples, and reaches up to $\sim 5 \times 10^7$ at the maximum
sampled time. The expected per-sample gate count is
$\mathbb{E}_{\tau}[r_{\rm RTE}(\tau)] = \mathbb{E}_{\tau}[\tau^2] \in
[4.0 \times 10^5,\, 1.4 \times 10^6]$. Both subroutines therefore demand
per-sample gate counts in the range $10^5$--$10^9$ for the typical and
worst-case Fourier samples, and combining either with the
$N_S \sim 10^8$ samples for which we have demonstrated the PF convergence
in \cref{fig:pf_estimator_convergence} brings into question the actual practicality of such randomized techniques.

\Cref{fig:fourier_time_distribution} also explains why we are unable to
present an analogous numerical study for the RTE estimator at the same
$\kappa$. The classical-simulation costs of the two estimators differ by
many orders of magnitude. The PF estimator's per-sample operator
$S_{r(\tau)}(\tau) = \big(S(\tau / r(\tau))\big)^{r(\tau)}$ is the
$r(\tau)$-th power of a deterministic factor $S(\tau / r(\tau))$ and can
be evaluated in $O(\log r(\tau))$ matrix multiplications by repeated
squaring. The RTE estimator's per-sample operator, by contrast, is a
product of $r_{\rm RTE}(\tau)$ \emph{independent} unitaries drawn from the
LCU distribution in \cref{eq:rte_lcu}, with no repeated factor to exploit;
classical evaluation requires $O(r_{\rm RTE}(\tau))$ matrix
multiplications per sample. With $N_S = 10^8$ samples per instance, the
total classical work to produce a counterpart to
\cref{fig:pf_estimator_convergence} for the RTE estimator scales as
$N_S \sum_{i=1}^{10} \mathbb{E}_{\tau}[\tau^2]_i \sim 10^{15}$ matrix
multiplications, six orders of magnitude beyond the $\sim 10^9$ matrix
multiplications required by the PF precomputation underlying
\cref{fig:pf_estimator_convergence}. This classical-simulation asymmetry
between $O(\log r)$ and $O(r)$ per-sample cost is a feature of how each
estimator is simulated and does not, on its own, reflect any fundamental
quantum-resource advantage of one method over the other.

To validate \cref{cor:rte_rmse} empirically and compare both subroutines on
the same matrix instances, we now repeat the experiment at a smaller
condition number $\kappa = 3$, where $t_{\max}$ is roughly two orders of
magnitude smaller and the per-sample RTE work $\mathbb{E}_\tau[\tau^2]$
falls into a classically tractable range. All other parameters are
unchanged from above: $\epsilon_F = 2 \times 10^{-3}$,
$\ket{\psi} = \ket{\phi} = \ket{0}$, $20$ Monte Carlo runs per $N_S$, and
$10$ random Hermitian matrices, with the latter now reused for both the
PF and the RTE estimator. Across the $10$ instances we find
$\lambda \in [1.32,\, 2.55]$, $t_{\max} \in [74,\, 154]$, and
$J K \in [1271,\, 3528]$.
 
For the PF estimator we retain the adaptive Trotter rule
\cref{eq:adaptive_pf_rule}. For the RTE estimator we use the
adaptive rule
\begin{align}
    \label{eq:adaptive_rte_rule}
    r_{\rm RTE}(\tau) := \max\!\left(1,\, \left\lceil c_{\rm RTE}\, \tau^2 \right\rceil\right),
\end{align}
which makes the per-sample LCU normalization
$\Nrte(\tau) = \alpha(\tau)^{r_{\rm RTE}(\tau)} \leq e^{\tau^2 / r_{\rm RTE}(\tau)} \leq e^{1/c_{\rm RTE}}$
uniform across the Fourier distribution and reduces the variance
prefactor in \cref{cor:rte_rmse} to the constant $\exp(2/c_{\rm RTE})$.
We fix $c_{\rm RTE} = 1$ for the empirical run, which carries the largest
variance prefactor ($\exp(2) \approx 7.39$) but the smallest per-sample
gate count. The truncation order $n_{\max}$ in \cref{eq:rte_finite} is
selected to be $n_{\max} = 4$ across instances. As with \cref{eq:max_analytical_bound_pf} for PF, we combine
\cref{cor:rte_rmse} with the
worst-case bias allocation $|B_{\rm RTE}| \leq \epsilon_F/2$ and the
Fourier error $\|F(A) - A^{-1}\| \leq \epsilon_F$ via the triangle
inequality, and take the max across the $10$ instances. 
 
\begin{figure}[h]
    \centering
    \includegraphics[width=0.8\linewidth]{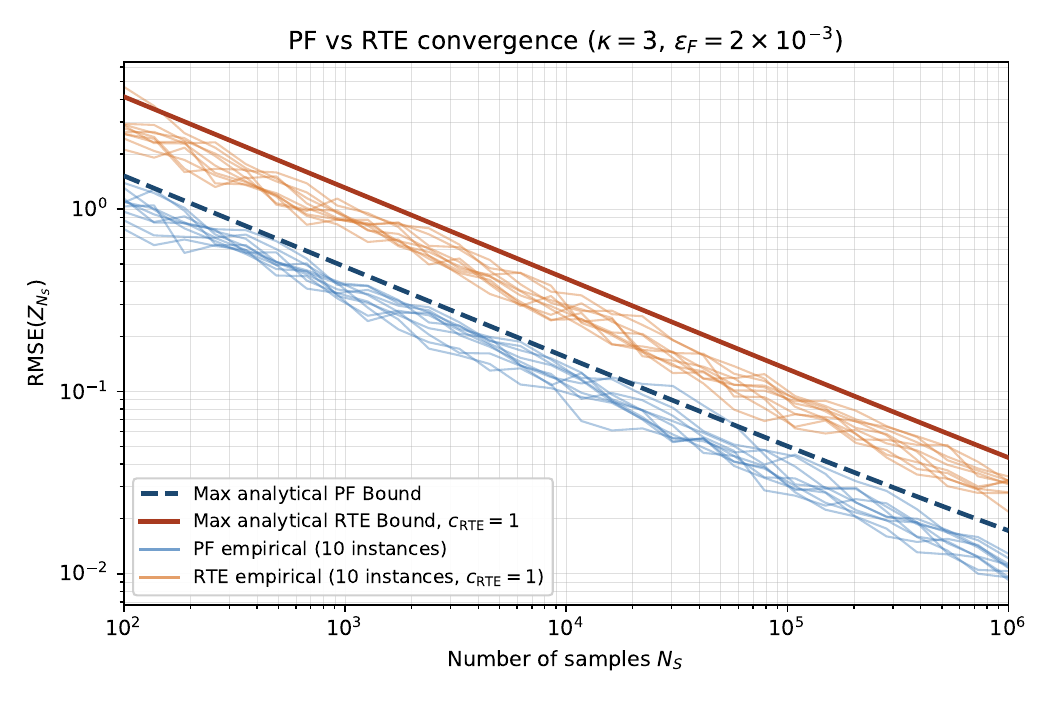}
    \caption{Empirical convergence of the PF and RTE estimators at condition number
    $\kappa = 3$ and an upper bound on the Fourier series error of $\epsilon_F = 2 \times 10^{-3}$ on the same $10$
    random Hermitian matrices. Thin blue curves show the empirical RMSE
    of the PF estimator under the adaptive rule \cref{eq:adaptive_pf_rule};
    thin orange curves show the empirical RMSE of the RTE estimator under
    \cref{eq:adaptive_rte_rule} with $c_{\rm RTE} = 1$. The dashed blue
    curve is the maximum across the $10$ instances of the analytical RMSE
    upper bound combining \cref{cor:trotter_rmse} with the Fourier
    approximation error of \cref{alg:fourier_series_construction} via the
    triangle inequality. The red curves are the analogous max-over-instances
    bounds for the RTE estimator from \cref{cor:rte_rmse} at three values
    of $c_{\rm RTE}$; the solid red curve corresponds to the
    $c_{\rm RTE} = 1$ used in the empirical run.}
    \label{fig:pf_rte_comparison}
\end{figure}
 
\Cref{fig:pf_rte_comparison} shows the empirical bundles which
lie below their respective max-over-instances bounds at nearly all measured points: only $4$ of $300$ measured PF points and $2$ of $300$
RTE points (at most $1.3\%$ in either family) exceed their respective
bounds, again due to the Monte-Carlo fluctuation of the empirical RMSE
estimator at $20$ runs per $N_S$ and our conservative choice for the single-shot Hadamard test. The mean empirical-to-bound ratios are
$0.73$ for PF and $0.69$ for RTE. We observe that in the variance-dominated regime both bounds decay as
$1/\sqrt{N_S}$, but the RTE bound at $c_{\rm RTE} = 1$ sits a factor
$\exp(1/c_{\rm RTE}) = e \approx 2.72$ above the PF bound; increasing
$c_{\rm RTE}$ tightens this gap (it falls as $\exp(1/c_{\rm RTE})$ and
approaches unity only as $c_{\rm RTE} \to \infty$) at the price of a
larger per-sample Trotter count $r_{\rm RTE}(\tau) = \lceil c_{\rm RTE}\,
\tau^2 \rceil$. The PF estimator therefore enjoys a strictly tighter
analytical bound at any finite $c_{\rm RTE}$. This statistical
advantage, however, must be weighed against the per-sample gate counts.

The per-instance $|\tau|$ distribution at $\kappa = 3$ is shown in
\cref{fig:fourier_time_distribution_kappa3}: it has the same qualitative
shape as \cref{fig:fourier_time_distribution} but at the smaller scale
set by $t_{\max} \in [74,\, 154]$, with median $|\tau| \in [9,\, 17]$ and
$99\%$-quantile $|\tau| \in [39,\, 78]$ across the $10$ instances. The
relevant moments are
$\mathbb{E}_\tau[r_{\rm PF}(\tau)] \in [3.7 \times 10^2,\, 2.2 \times 10^3]$
and $\mathbb{E}_\tau[\tau^2] \in [1.97 \times 10^2,\, 7.85 \times 10^2]$.
With $L = 16$ Pauli terms per matrix, the per-sample expected gate
counts are
$\mathbb{E}_\tau[N_{\rm CP}^{\rm PF}] = 2 L\, \mathbb{E}_\tau[r_{\rm PF}]
\in [1.2 \times 10^4,\, 7.0 \times 10^4]$ for PF and
$\mathbb{E}_\tau[N_{\rm CP}^{\rm RTE}] = c_{\rm RTE}\, \mathbb{E}_\tau[\tau^2]
\in [2.0 \times 10^2,\, 7.9 \times 10^2]$ for RTE at $c_{\rm RTE} = 1$.
RTE is therefore roughly two orders of magnitude cheaper per sample at
this $\kappa$ and $L$.

\begin{figure}[h]
    \centering
    \includegraphics[width=0.8\linewidth]{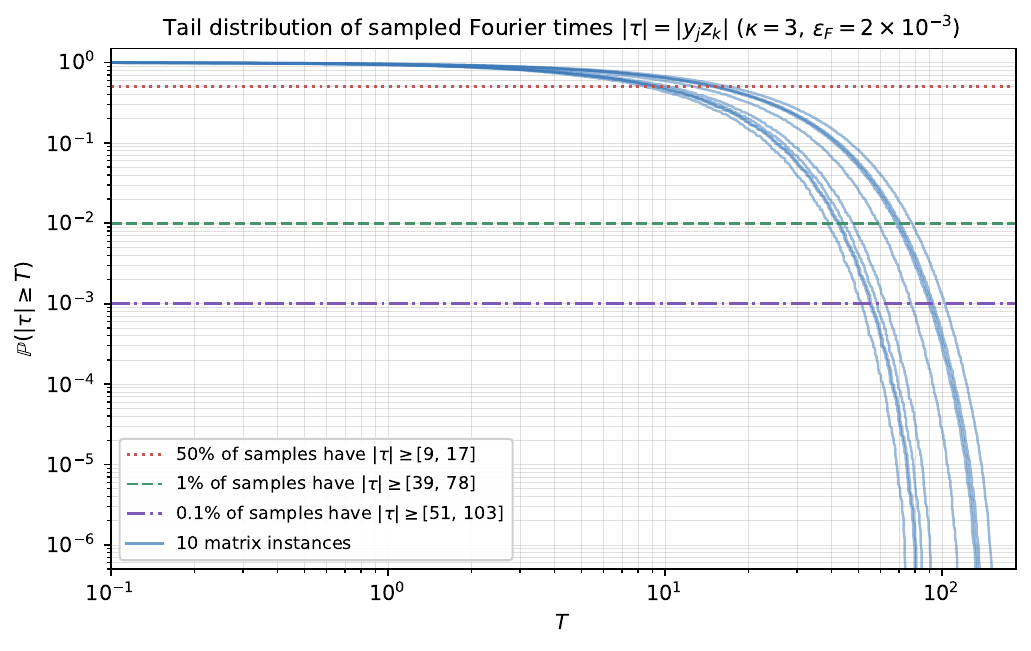}
    \caption{Tail distribution of the absolute value of the sampled Fourier times
    $|\tau|$ under the joint sampling distribution
    $p_y \otimes p_z$ from \cref{eq:probs}, for the same $10$ random
    Hermitian matrices at $\kappa = 3$ as used in
    \cref{fig:pf_rte_comparison}. Each thin curve gives the complementary
    cumulative distribution function $\mathbb{P}(|\tau| \geq T)$ for one
    matrix instance, and the horizontal lines mark the corresponding
    ranges of the $50\%$, $1\%$, and $0.1\%$ quantiles of $|\tau|$ across
    the $10$ instances. The maximum sampled time per instance ranges over
    $t_{\max} \in [74,\, 154]$, two orders of magnitude smaller than at
    $\kappa = 100$ in \cref{fig:fourier_time_distribution}.}
    \label{fig:fourier_time_distribution_kappa3}
\end{figure}
 
Combining the two effects, the total quantum-gate budget required to
reach a target RMSE $\epsilon$ in the variance-dominated regime is the
product of the sample count $N_S$ needed to drive the variance below
$\epsilon^2$ and the expected per-sample gate count of the chosen
subroutine. For PF, $N_S^{\rm PF} \propto 2 (N_y N_z)^2 / (\lambda^2 \epsilon^2)$;
for RTE at $c_{\rm RTE}$,
$N_S^{\rm RTE} \propto 2\, e^{2/c_{\rm RTE}} (N_y N_z)^2 / (\lambda^2 \epsilon^2)$,
larger by the factor $e^{2/c_{\rm RTE}}$ that captures RTE's variance
penalty (a factor of $e^2 \approx 7.4$ at $c_{\rm RTE} = 1$). Defining
the total gate budgets
$G_{\rm PF} := N_S^{\rm PF} \cdot \mathbb{E}_\tau[N_{\rm CP}^{\rm PF}(\tau)]$ and
$G_{\rm RTE}(c_{\rm RTE}) := N_S^{\rm RTE} \cdot \mathbb{E}_\tau[N_{\rm CP}^{\rm RTE}(\tau)]$,
this yields, up to instance-dependent prefactors of order unity,
\begin{align}
    G_{\rm PF} \;&\propto\; \frac{2 (N_y N_z)^2}{\lambda^2 \epsilon^2} \cdot 2 L\, \mathbb{E}_\tau[r_{\rm PF}(\tau)],
    \nonumber \\
    G_{\rm RTE}(c_{\rm RTE}) \;&\propto\; \frac{2\, e^{2/c_{\rm RTE}} (N_y N_z)^2}{\lambda^2 \epsilon^2} \cdot c_{\rm RTE}\, \mathbb{E}_\tau[\tau^2],
    \label{eq:gate_budgets}
\end{align}
with ratio
\begin{align}
    \frac{G_{\rm RTE}(c_{\rm RTE})}{G_{\rm PF}}
    \;=\; e^{2/c_{\rm RTE}} \cdot \frac{c_{\rm RTE}\, \mathbb{E}_\tau[\tau^2]}{2 L\, \mathbb{E}_\tau[r_{\rm PF}(\tau)]}.
    \label{eq:gate_budget_ratio}
\end{align}
The two factors in \cref{eq:gate_budget_ratio} pull in opposite
directions: $e^{2/c_{\rm RTE}}$ is the sample-count penalty paid by RTE
for its larger variance and exceeds unity at any finite $c_{\rm RTE}$,
while $c_{\rm RTE} \mathbb{E}_\tau[\tau^2] / (2 L\, \mathbb{E}_\tau[r_{\rm PF}(\tau)])$
is the per-sample gate-count ratio and, on our instances, is much smaller
than unity ($\sim 1/85$ at $c_{\rm RTE} = 1$). The product $e^{2/c_{\rm RTE}} \cdot c_{\rm RTE}$ is minimized
at $c_{\rm RTE}^\star = 2$. Numerically on our $10$ instances, this optimum yields
$G_{\rm RTE}/G_{\rm PF} \approx 0.06$, while the empirical
$c_{\rm RTE} = 1$ yields $G_{\rm RTE}/G_{\rm PF} \approx 0.09$. The RTE
subroutine therefore reaches a given target RMSE with roughly an order
of magnitude smaller total gate budget than PF on these instances: the
per-sample gate-count saving more than compensates for the larger sample
count.
 
This conclusion strengthens with system size. The PF per-sample cost
$N_{\rm CP}^{\rm PF}(\tau) = 2 L\, r_{\rm PF}(\tau)$ scales linearly in
$L$. The dependence on $L$ enters only through the explicit factor
$2L$: the factor $r_{\rm PF}(\tau) \propto \sqrt{f}\, \tau^{3/2}$
carries no $L$-dependence, since the Trotter constant $f$ of
\cref{lemma:trotter_number} is $\mathcal{O}(1)$, and in particular independent
of both $L$ and $\lambda$, for the $\lambda$-rescaled Hamiltonian
$\tilde A = \sum_\ell \tilde P_\ell$. This follows because the partial sums of
the rescaled Paulis are controlled by $\sum_\ell \|\tilde P_\ell\| = 1$, so the
nested-commutator norms defining $f$ (cubic and quadratic in the
$\tilde P_\ell$) admit a bound independent of the number of terms. A direct
estimate gives $f \leq 5/12$. The naive expectation that $f$ grows with $L$
applies to an unrescaled Hamiltonian with $\mathcal{O}(1)$-norm terms and does
not survive the $\ell_1$-rescaling $\tilde P_\ell = c_\ell P_\ell / \lambda$.
The per-sample RTE cost $N_{\rm CP}^{\rm RTE}(\tau) = c_{\rm RTE}\, \tau^2$, by
contrast, has no $L$-dependence at all: each LCU segment is a single controlled
Pauli rotation regardless of how many Pauli terms appear in the decomposition.
The total-gate ratio \cref{eq:gate_budget_ratio} therefore favors RTE by an
additional factor of $L$ in this regime, strengthening the quantitative case
for RTE at the larger $L$ of physically relevant problems---though our
experiments reveal that absolute gate budgets at any practical $\kappa$ remain
prohibitive on near-term hardware.

Finally, we comment on the tightness of our bounds. The empirical RMSE agrees with the analytical RMSE
bound to a factor of order unity across all instances. Taken together, the Fourier-approximation bound is loose by
a controlled, dimension-independent factor (the empirical error is $17$--$21\%$
of $\epsilon_F$) and the fact that worst-case
Trotter bound is known to be pessimistic, in practice explain the slack in our bounds.

\section{Discussion}
\label{sec:discussion}

In this work we have provided non-asymptotic, end-to-end resource
estimates for a randomized hybrid quantum algorithm that estimates
scalar properties of matrix functions $\bra{\phi} f(A) \ket{\psi}$,
specializing to the linear systems case $f(A) = A^{-1}$. Three
contributions distinguish our analysis from prior asymptotic treatments.
We provide explicit Fourier truncation and discretization
parameters $(y_{\max}, z_{\max}, J, K)$ that achieve a prescribed spectral-norm
error $\epsilon_F$ in exponentially fewer quadrature nodes for the oscillatory
$y$-integral than a low-order discretization, by combining the Gauss-Legendre rule
of~\cite{qrt24} for that integral with a trapezoidal rule for the rapidly decaying
$z$-integral. We also provide explicit
sample-count and per-sample gate-count formulas for two
Hamiltonian-simulation subroutines---a second-order product formula
(PF) and the random Taylor expansion (RTE) of~\cite{Wan:2021non,
Wang:2023els}---in both high-probability and root-mean-squared-error
convergence modes. Our numerical results validate these bounds
end-to-end on small random Hermitian matrices and, at a smaller
condition number where both subroutines are classically tractable,
place PF and RTE on the same instances to expose their relative
quantum-resource trade-off.

Two structural features dominate the resulting resource picture. First,
for both subroutines and both convergence modes the sample complexity
scales as $\widetilde{\mathcal{O}}(\kappa^2/\epsilon^2)$
(\cref{tab:resource_summary}), with the prefactor governed by
$N_y N_z / \lambda$---a quantity that grows linearly in the rescaled
condition number $\tilde\kappa^* = \lambda \kappa$ up to logarithmic
factors (\cref{lemma:normalization_bounds}) and that embodies the cost
of Fourier-coefficient sampling. Second, the per-sample non-Clifford
gate count $N_{\rm CP}$ depends asymmetrically on the subroutine: PF
carries an explicit factor $2L$ from the second-order Trotter formula
together with a $\sqrt{f}\, \tau^{3/2}$ scaling from the Trotter-error
bound, where the Trotter constant $f$ of \cref{lemma:trotter_number}
is $\mathcal{O}(1)$ for the $\lambda$-rescaled Hamiltonian
$\tilde A = \sum_\ell \tilde P_\ell$ (the commutator sum is controlled by
$\sum_\ell \|\tilde P_\ell\| = 1$ and is therefore independent of both $L$
and $\lambda$), so that the $L$-dependence of $N_{\rm CP}^{\rm PF}$ is the
single explicit power $L^1$. RTE, by contrast, carries no explicit
$L$-dependence and only a $\tau^2$ scaling under the choice
$r_{\rm RTE}(\tau) \propto \tau^2$ that uniformizes the variance prefactor.
Together these features place both
subroutines far outside near-term hardware reach at any condition number
of practical interest: even a $4 \times 4$ instance with $\kappa = 100$
and target error $\epsilon \sim 10^{-3}$ requires on the order of
$10^{15}$ controlled Pauli rotations in total to reach convergence.

The side-by-side $\kappa = 3$ comparison in our numerical results refines the picture in a direction
that the explicit non-asymptotic prefactors derived in this work make
quantitatively visible. RTE's secondary sampling step over the LCU
expansion of the time evolution operator inflates the variance prefactor
in \cref{cor:rte_rmse} by $\exp(2/c_{\rm RTE})$, demanding roughly
$7\times$ more samples than PF at $c_{\rm RTE} = 1$ to reach the same
target error. On the same matrix instances, however, RTE is roughly
$85\times$ cheaper per sample, since each LCU segment contributes a
single controlled Pauli rotation independent of the Pauli sparsity $L$.
The two effects combine into a total gate budget ratio
$G_{\rm RTE} / G_{\rm PF} \approx 0.06$ at the optimum
$c_{\rm RTE}^\star = 2$ and $\approx 0.09$ at $c_{\rm RTE} = 1$, with
the advantage scaling as $L$ in the worst case, inherited from the
explicit $2L$ factor of the product formula. The RTE subroutine
is therefore the more efficient of the two on a quantum device---a
relative ranking that the explicit non-asymptotic bounds derived in
this work surface, but that asymptotic complexity statements alone could
not have determined.

Despite this relative ordering, the absolute gate budgets remain
prohibitive at any $\kappa$ of practical interest. The bottleneck is primarily the fact that the maximum
simulation time $t_{\max}$ scales as $\kappa \log(\kappa/\epsilon_T)$,
the Fourier grid size $JK$ as a higher polynomial in $\kappa$, and the
sampling normalization $N_y N_z$ inherits the same $\kappa$-dependence. We note that the precise $\kappa^4$ scaling is tied to this particular Fourier
construction, whereas the $1/\epsilon^2$ sampling cost and the variance's
dependence on the squared $\ell_1$-normalization of the Fourier series are generic
to any Fourier-sampling randomized scheme; the polynomial-in-$1/\epsilon$ cost is
therefore structural, while the power of $\kappa$ is construction-dependent.
Several avenues could improve this picture. First, the
Gauss-Legendre--trapezoidal construction we adopt from~\cite{qrt24} is
only one possible discretization; alternative integral representations
of the inverse function with weaker $\kappa$-dependence in $t_{\max}$
would directly reduce the per-sample Trotter or LCU work for both
subroutines. Second, classical preconditioning of $A$, as suggested
in~\cite{Wang:2023els}, addresses the $\kappa$-dependence at its source
by reducing the effective condition number prior to applying the
quantum routine, at the cost of a one-time classical computation.
Third, lower-variance variants of the RTE construction---for instance,
Hamiltonian-aware sampling distributions over the LCU rather than the
importance distribution we use---could close the $\exp(2/c_{\rm RTE})$
gap to PF's variance prefactor and further widen RTE's lead. Fourth,
structured Hamiltonians with locality or sparsity admit Trotter
constants $f$ even smaller than the worst-case bound of
\cref{lemma:trotter_number}, and a sampling procedure that respects
such structure could further reduce both per-sample costs.

One natural alternative against which to weigh these costs is the family
of block-encoding-based linear systems solvers, which achieve query
complexity $\widetilde{\mathcal{O}}(\kappa \log(1/\epsilon))$ with respect
to the block-encoding access via the quantum singular value
transformation~\cite{qsvt} and its variants, by Chebyshev approximation
of the inverse function over $\domain_{\kappa^*}$
\cite{doi:10.1137/1.9781611975949}. We emphasize first that this
comparison does not reflect any advantage in the approximation step
itself: the Fourier series we construct already attains a maximum
evolution time $t_{\max} = \mathcal{O}(\kappa\log(\kappa/\epsilon))$
(\cref{eq:max_fourier_time}), matching the degree
$\mathcal{O}(\kappa\log(\kappa/\epsilon))$ of the Chebyshev approximation,
with both linear in $\kappa$ and logarithmic in $1/\epsilon$ at the level
of the approximation. The apparent disparity in cost instead originates in
how each approximation is realized: QSVT implements a degree-$d$ polynomial
using $\mathcal{O}(d)$ coherent queries to the block encoding, whereas the
randomized scheme realizes a Fourier approximation of comparable degree
through $N_{\rm S} = \widetilde{\mathcal{O}}(\kappa^2/\epsilon^2)$
incoherent Monte Carlo samples, each requiring its own per-sample
Hamiltonian simulation. The polynomial $1/\epsilon$-dependence of the
randomized total cost and the overall $\kappa^4$ scaling therefore arise
from the sampling and per-sample simulation overheads rather than from the
approximation degree.

Two caveats, however, prevent this contrast from being read as a settled
verdict. First, the two
schemes do not return the same object. The randomized scheme produces a
\emph{scalar} estimate of $\bra{\phi} A^{-1} \ket{\psi}$ already resolved to
additive precision $\epsilon$, and its $1/\epsilon^2$ cost reflects exactly
that. The QSVT query complexity $\widetilde{\mathcal{O}}(\kappa\log(1/\epsilon))$,
by contrast, prepares the solution \emph{state}
$A^{-1}\ket{\psi}/\lVert A^{-1}\ket{\psi}\rVert$; recovering the same scalar
$\bra{\phi}A^{-1}\ket{\psi}$ from that state requires an additional overlap
estimation against $\ket{\phi}$, whose cost grows polynomially in $1/\epsilon$
(linearly under coherent amplitude estimation, quadratically under direct
sampling). A like-for-like, scalar-to-scalar comparison therefore loads a
$\mathrm{poly}(1/\epsilon)$ factor onto the block-encoding side that its bare
query complexity does not display, partially closing the gap on the
$1/\epsilon$ axis. Second, the query complexity counts calls to the block
encoding as unit-cost. A single block-encoding query is itself a quantum
circuit whose gate cost depends on the construction---for linear-combination-of-unitaries
encodings it typically scales with the coefficient one-norm $\lambda$, and the
subnormalization of the encoding feeds back into the effective condition number
seen by QSVT---so a fully gate-level accounting can reintroduce $\kappa$- and
$\lambda$-dependence that the query count suppresses. Whether, after both
adjustments, the logarithmic $1/\epsilon$-dependence of the block-encoding
query complexity still translates into a genuine gate-level advantage over the
polynomial $1/\epsilon$-dependence of the randomized total cost is a question
our explicit non-asymptotic bounds make well-posed but which we do not settle
here; it is especially pertinent on the structured, data-sparse matrices for
which efficient block encodings are
known~\cite{doi:10.1137/22M1484298, Sunderhauf2024blockencoding,
Kane:2024odt, Hariprakash:2023tla, Rhodes:2024zbr, LIU2025102480}, and we leave
a head-to-head gate-level comparison on such instances to future work. The
randomized approach retains, in any case, the appealing features of
$\log(N) + 1$ qubit width and structure-agnostic applicability, but the
explicit gate budgets we have computed make clear that those features alone do
not suffice to recommend the method at problem sizes of scientific interest.

\section*{Data Availability}
Numerical data and python notebooks capable of reproducing the figures in this work can be found here: \url{https://github.com/sid-hari/RQA_paper_plots}

\section*{Code Availability}

Code used to generate data in this study is available from the corresponding authors upon reasonable request

\section*{Acknowledgements}

The authors thank Niel Van Buggenhout and Yizhi Shen for helpful discussions.
This research was supported by the U.S. Department of Energy (DOE) Office of Science (SC) under Contract No. DEAC02-05CH11231, through the National Energy Research Scientific Computing Center (NERSC), a User Facility located at Lawrence Berkeley National Laboratory, using NERSC award ASCR-ERCAP0033494. S.H.\ acknowledges funding through the Quantum Information Science Enabled Discovery (QuantISED) for High Energy Physics (KA2401032) program funded by DOE SC and the Quantum Systems Accelerator (QSA), a DOE SC National Quantum Information Science Research Center. R.V.B.\ acknowledges funding through the DOE SC Accelerated Research in Quantum Computing, Fundamental Algorithmic Research toward Quantum Utility (FAR-Qu) program. 

\section*{Author Contributions}

S.H. developed the theory, performed the numerical simulations presented in this work, and wrote the initial draft. R.V.B., K.K. and D.C. helped supervise the projects and refine the ideas. D.C. proposed the scope of the work. All authors discussed the results and contributed to the final manuscript.

\section*{Competing Interests}

All authors declare no financial or non-financial competing interests.

\printbibliography

\clearpage

\setcounter{secnumdepth}{1}
\setcounter{section}{0}
\setcounter{equation}{0}
\setcounter{figure}{0}
\setcounter{table}{0}
\renewcommand{\thesection}{\arabic{section}}
\renewcommand{\theequation}{S\arabic{equation}}
\renewcommand{\thefigure}{S\arabic{figure}}
\renewcommand{\thetable}{S\arabic{table}}

\begin{center}
  {\Large\bfseries Supplementary Information}\\[0.4em]
  {\large\itshape The Practicality of Randomized Quantum Linear Systems Solvers}
\end{center}
\vspace{1.5em}

\section{Choosing the number of shots per Hadamard test}
\label{app:1_shot_Hadamard}

Suppose we wish to estimate
\begin{equation}
    F(A) := \bra{0}U_\psi^{\dagger} f(A) U_\phi\ket{0} = \sum_{s\in[S']} \alpha_s\, f_s, \qquad f_s := \bra{0}U_\psi^{\dagger} e^{-iAt_s} U_\phi \ket{0},
\end{equation}
where the amplitudes $f_s \in \mathbb{C}$ are estimated via Hadamard tests and the Fourier coefficients $\alpha_s \in \mathbb{C}$ may also be complex. We use the Monte Carlo estimator
\begin{equation}
    \hat{Z} = \frac{1}{n} \sum_{i\in[n]} \frac{\alpha_{s_i}}{p_{s_i}}\, \hat{f}_{s_i},
\end{equation}
where each $s_i$ is drawn independently with probability $p_s = |\alpha_s| / C$ (with $C := \sum_s |\alpha_s|$), and $\hat{f}_{s_i}$ is an unbiased estimator of $f_{s_i}$ obtained from $S$ Hadamard shots for the real part and $S$ shots for the imaginary part. Each individual Hadamard shot returns a $\pm 1$ outcome with mean equal to $\Re f_s$ or $\Im f_s$ and variance at most $1$, so $\mathrm{Var}(\hat f_s) = \sigma_s^2/S$ for some $\sigma_s^2 \leq 2$ that depends only on $|f_s|$.

For complex random variables we use the standard convention $\mathrm{Var}(Z) := \mathbb{E}[|Z|^2] - |\mathbb{E}[Z]|^2$. Conditioning on the sampled index $s$ and applying the law of total variance,
\begin{align}
    \mathrm{Var}[\hat{Z}] &= \frac{1}{n}\, \mathrm{Var}\!\left(\frac{\alpha_s}{p_s}\, \hat{f}_s\right) \\
    &= \frac{1}{n}\!\left(\, \mathbb{E}_s\!\left[\frac{|\alpha_s|^2}{p_s^2}\, \mathrm{Var}(\hat{f}_s)\right] + \mathrm{Var}_s\!\left(\frac{\alpha_s}{p_s}\, f_s\right) \right).
\end{align}
Substituting $p_s = |\alpha_s|/C$ and $\mathrm{Var}(\hat{f}_s) = \sigma_s^2/S$, the inner expectation evaluates to
\begin{equation}
    \mathbb{E}_s\!\left[\frac{|\alpha_s|^2}{p_s^2}\, \frac{\sigma_s^2}{S}\right] = \frac{1}{S}\sum_s p_s\, \frac{|\alpha_s|^2}{p_s^2}\, \sigma_s^2 = \frac{C}{S} \sum_s |\alpha_s|\, \sigma_s^2,
\end{equation}
and the second term evaluates to
\begin{equation}
    \mathrm{Var}_s\!\left(\frac{\alpha_s}{p_s}\, f_s\right) = \sum_s p_s\, \frac{|\alpha_s|^2 |f_s|^2}{p_s^2} - |F(A)|^2 = C \sum_s |\alpha_s|\, |f_s|^2 - |F(A)|^2.
\end{equation}
Combining these,
\begin{equation}
    \mathrm{Var}[\hat{Z}] = \frac{1}{n}\!\left( \frac{C}{S} \sum_s |\alpha_s|\, \sigma_s^2 \;+\; C \sum_s |\alpha_s|\, |f_s|^2 - |F(A)|^2 \right).
\end{equation}

Suppose now that we have a fixed total budget $T := nS$ of paired Hadamard test evaluations, where each evaluation consists of one real-part shot and one imaginary-part shot. Substituting $n = T/S$,
\begin{align}
    \mathrm{Var}[\hat{Z}] = \frac{C}{T} \sum_s |\alpha_s|\, \sigma_s^2 \;+\; \frac{S}{T}\, \mathrm{Var}_s\!\left(\frac{\alpha_s}{p_s}\, f_s\right).
\end{align}
The first term is independent of $S$, while the second is non-negative and grows linearly in $S$. The variance is therefore minimized at the smallest allowed value, $S = 1$ (with $n = T$), corresponding to a single Hadamard pair per Fourier sample.

\section{Proof of Theorem 4}
\label{app:thm3}

\begin{proof}
We use a result from the proof of Lemma 12 in~\cite{AndrewM:2015yvx}, which establishes that for all $x \neq 0$,
\begin{equation}
\label{eq:cks_bound}
\left| I(x) - \frac{\imi}{\sqrt{2\pi}} \int_{0}^{y_{\max}} dy \int_{-z_{\max}}^{z_{\max}} dz \, z e^{-z^2/2} e^{-\imi x y z} \right|
\leq \frac{e^{-(x y_{\max})^2/2}}{|x|} + \frac{2 e^{-z_{\max}^2/2}}{|x|}.
\end{equation}
We bound the two terms on the right-hand side by $\epsilon_1$ and $\epsilon_2$ respectively, requiring $\epsilon_1 + \epsilon_2 \leq \epsilon_T$ so that the total is at most $\epsilon_T$. For $x \in \domain_\kappa$ we have $1/\kappa \leq |x| \leq 1$, hence $1/|x| \leq \kappa$.

\emph{First term.} Since $e^{-(x y_{\max})^2 /2}$ is maximized over $\domain_\kappa$ at $|x| = 1/\kappa$,
\begin{equation}
\frac{e^{-(x y_{\max})^2/2}}{|x|} \leq \kappa \, e^{-y_{\max}^2 / (2\kappa^2)}.
\end{equation}
Demanding $\kappa \, e^{-y_{\max}^2 / (2\kappa^2)} \leq \epsilon_1$ yields
\begin{equation}
y_{\max} \geq \kappa \sqrt{2 \ln(\kappa/\epsilon_1)}.
\end{equation}

\emph{Second term.} Using $1/|x| \leq \kappa$,
\begin{equation}
\frac{2 e^{-z_{\max}^2/2}}{|x|} \leq 2\kappa \, e^{-z_{\max}^2/2}.
\end{equation}
Demanding $2\kappa \, e^{-z_{\max}^2/2} \leq \epsilon_2$ yields
\begin{equation}
z_{\max} \geq \sqrt{2 \ln(2\kappa/\epsilon_2)}.
\end{equation}

The choice $\epsilon_1 = \epsilon_T / 3$, $\epsilon_2 = 2\epsilon_T / 3$ saturates $\epsilon_1 + \epsilon_2 = \epsilon_T$ and renders the two logarithmic arguments equal,
\begin{equation*}
\frac{\kappa}{\epsilon_1} \;=\; \frac{2\kappa}{\epsilon_2} \;=\; \frac{3\kappa}{\epsilon_T},
\end{equation*}
producing the stated formulas for $y_{\max}$ and $z_{\max}$.
\end{proof}

\section{Proof of Theorem 5}
\label{app:thm4}
\begin{proof}
Throughout the proof we work with the rescaled form \cref{eq:GL_transformed_integral} of the truncated integral and write the discretization error as
$\epsilon_D := \sup_{x \in \domain_\kappa} | I_D^{(\epsilon_D,\epsilon_T)}(x) - I_T^{(\epsilon_T)}(x)|$.
To separate the contributions of the trapezoidal rule (in $z$) and the Gauss--Legendre rule (in $y$), we introduce an intermediate object in which only the $z$-integral has been discretized,
\begin{equation}
I_{T,D_z}(x) := \frac{\imi y_{\max}}{2\sqrt{2\pi}} \int_{-1}^{1} dy \sum_{k \in [K]} w_z^{(k)} z_k e^{-z_k^2/2} e^{-\imi x y_{\max}(1+y)z_k/2}.
\end{equation}
The triangle inequality then yields
\begin{equation}
\label{eq:disc_decomposition}
| I_D^{(\epsilon_D,\epsilon_T)}(x) - I_T^{(\epsilon_T)}(x)| \leq \underbrace{| I_{T,D_z}(x) - I_T^{(\epsilon_T)}(x)|}_{=:\epsilon_{(D,\mathrm{trap})}(x)} + \underbrace{| I_D^{(\epsilon_D,\epsilon_T)}(x) - I_{T,D_z}(x)|}_{=:\epsilon_{(D,\mathrm{GL})}(x)},
\end{equation}
and we proceed by bounding the two contributions separately, repeatedly using $|x| \leq 1$ for $x \in \domain_\kappa$.

We begin with the trapezoidal contribution. For a function $g$ analytic on the horizontal strip $\mathcal{H}_\rho = \{w \in \mathbb{C} : |\Im w| \leq \rho\}$, the trapezoidal rule on $(-\infty, \infty)$ with step $\Delta z$ satisfies~\cite{trefethen14}
\begin{equation}
\label{eq:trap_error_bound}
e_{\mathrm{trap}}^{(\Delta z)}[g \mid \mathcal{H}_\rho] \leq \frac{2 \sup_{|s| \leq \rho} \int_{-\infty}^{\infty} |g(z + \imi s)|\, dz}{e^{2\pi\rho/\Delta z} - 1}.
\end{equation}
Fixing $x \in \domain_\kappa$ and $y \in [-1, 1]$ and setting $g_{x,y}(z) := z e^{-z^2/2} e^{-\imi x y_{\max}(1+y) z / 2}$, the integrand is entire and so \cref{eq:trap_error_bound} applies for any $\rho > 0$. A direct computation gives
\begin{equation}
|g_{x,y}(z + \imi s)| = |z + \imi s| \, e^{-z^2/2 + s^2/2} \, e^{x s y_{\max}(1+y)/2},
\end{equation}
and using $\int |z|e^{-z^2/2}dz = 2$ together with $\int e^{-z^2/2} dz = \sqrt{2\pi}$ yields
\begin{equation}
\int_{-\infty}^{\infty} |g_{x,y}(z + \imi s)|\, dz \leq e^{s^2/2} e^{x s y_{\max}(1+y)/2} \left( 2 + |s|\sqrt{2\pi} \right).
\end{equation}
Taking the supremum over $|s| \leq \rho$, while recalling $|x| \leq 1$ and $0 \leq (1+y)/2 \leq 1$, we obtain
\begin{equation}
\sup_{|s| \leq \rho} \int |g_{x,y}(z + \imi s)|\, dz \leq e^{\rho^2/2} e^{|x| y_{\max}(1+y) \rho/2} \left(2 + \rho\sqrt{2\pi}\right).
\end{equation}
Substituting $\Delta z = 2 z_{\max}/(K-1)$ into \cref{eq:trap_error_bound}, integrating over $y \in [-1,1]$, and applying the elementary identity $\int_{-1}^{1} e^{a(1+y)/2} dy = 2(e^a - 1)/a$ gives
\begin{align}
\epsilon_{(D,\mathrm{trap})}(x)
&\leq \frac{y_{\max}}{2\sqrt{2\pi}} \int_{-1}^{1} \frac{2 e^{\rho^2/2}(2 + \rho\sqrt{2\pi}) e^{|x| y_{\max}(1+y)\rho / 2}}{e^{\pi \rho (K-1)/z_{\max}} - 1}\, dy \nonumber \\
&= \frac{2 e^{\rho^2/2} (2 + \rho\sqrt{2\pi}) (e^{|x| y_{\max}\rho} - 1)}{|x|\rho \sqrt{2\pi} (e^{\pi \rho (K-1)/z_{\max}} - 1)}.
\end{align}
Since $u \mapsto (e^{a u} - 1)/u$ is increasing in $u$ for $a > 0$, the right-hand side is monotonically increasing in $|x| \in [1/\kappa, 1]$, so taking the supremum over $\domain_\kappa$ gives
\begin{equation}
\label{eq:trap_bound_supx}
\sup_{x \in \domain_\kappa} \epsilon_{(D,\mathrm{trap})}(x) \leq \frac{2 e^{\rho^2/2} (2 + \rho\sqrt{2\pi}) (e^{y_{\max}\rho} - 1)}{\rho \sqrt{2\pi} (e^{\pi \rho (K-1)/z_{\max}} - 1)}.
\end{equation}
To extract a closed-form lower bound on $K$, we invoke the elementary inequality
\begin{equation}
\label{eq:exp_ratio}
\frac{e^A - 1}{e^B - 1} \leq e^{A - B} \quad \text{whenever } 0 \leq A \leq B,
\end{equation}
which follows from $(e^A - 1) e^B \leq (e^B - 1) e^A \iff e^A \leq e^B$. We will verify at the end of the proof that the choice of $K$ derived below guarantees $y_{\max} \rho \leq \pi \rho (K - 1)/z_{\max}$, justifying the use of \cref{eq:exp_ratio} with $A = y_{\max} \rho$ and $B = \pi \rho (K-1)/z_{\max}$. Under this condition, \cref{eq:trap_bound_supx} simplifies to
\begin{equation}
\sup_{x \in \domain_\kappa} \epsilon_{(D,\mathrm{trap})}(x) \leq \frac{2 e^{\rho^2/2} (2 + \rho\sqrt{2\pi})}{\rho \sqrt{2\pi}} \, e^{y_{\max} \rho - \pi \rho (K - 1)/z_{\max}},
\end{equation}
and demanding the right-hand side be at most $\epsilon'$, taking logarithms, and rearranging produces
\begin{equation}
\label{eq:K_general}
K \geq 1 + \frac{z_{\max}}{\pi \rho} \left[ \rho y_{\max} + \ln 2 + \frac{\rho^2}{2} + \ln\!\left(1 + \frac{2}{\rho \sqrt{2\pi}}\right) + \ln\frac{1}{\epsilon'} \right].
\end{equation}

We now turn to the Gauss--Legendre contribution. For a function $h$ analytic on and within the Bernstein $\sigma$-ellipse $\mathcal{E}_\sigma = \{(\sigma e^{\imi\theta} + \sigma^{-1} e^{-\imi\theta})/2 : \theta \in [-\pi, \pi]\}$, the Gauss--Legendre rule on $[-1,1]$ with $J$ nodes satisfies~\cite{qrt24}
\begin{equation}
\label{eq:gl_error_bound}
e_{\mathrm{GL}}^{(J)}[h \mid \mathcal{E}_\sigma] \leq \frac{8 \sigma^2 \sup_{t \in \mathcal{E}_\sigma} |h(t)|}{\sigma^{2J}(\sigma^2 - 1)}.
\end{equation}
Writing both $I_{T,D_z}$ and $I_D^{(\epsilon_D,\epsilon_T)}$ as a sum over the trapezoidal nodes $z_k$ and exploiting linearity of the GL error,
\begin{equation}
\epsilon_{(D,\mathrm{GL})}(x) \leq \frac{y_{\max}}{2\sqrt{2\pi}} \sum_{k \in [K]} w_z^{(k)} |z_k| e^{-z_k^2/2} \cdot e_{\mathrm{GL}}^{(J)}[h_{x,k} \mid \mathcal{E}_\sigma],
\end{equation}
where $h_{x,k}(y) := e^{-\imi x y_{\max}(1+y)z_k/2}$. This integrand is again entire, so \cref{eq:gl_error_bound} applies for any $\sigma > 1$. To bound $\sup_{t \in \mathcal{E}_\sigma}|h_{x,k}(t)|$, we use the parametrization $t = \frac{\sigma + \sigma^{-1}}{2} \cos\theta + \imi\, \frac{\sigma - \sigma^{-1}}{2} \sin\theta$ for $\mathcal{E}_\sigma$, from which $|\Im(t)| \leq (\sigma - \sigma^{-1})/2$. Since $x, z_k, y_{\max}$ are real, we have $|h_{x,k}(t)| = e^{x z_k y_{\max} \Im(t) / 2}$, which is maximized when $\Im(t)$ takes its extreme value with the appropriate sign. Hence, using $|x| \leq 1$,
\begin{equation}
\label{eq:bernstein_sup}
\sup_{t \in \mathcal{E}_\sigma} |h_{x,k}(t)| \leq e^{|z_k| y_{\max} (\sigma - \sigma^{-1})/4}.
\end{equation}
Combining \cref{eq:gl_error_bound,eq:bernstein_sup} and bounding $|z_k| \leq z_{\max}$, $e^{-z_k^2/2} \leq 1$, and $w_z^{(k)} = \Delta z = 2 z_{\max}/(K-1)$, we find
\begin{align}
\sup_{x \in \domain_\kappa} \epsilon_{(D,\mathrm{GL})}(x)
&\leq \frac{y_{\max}}{2\sqrt{2\pi}} \cdot \frac{2 z_{\max}}{K - 1} \cdot K z_{\max} \cdot \frac{8 \sigma^2 \, e^{z_{\max} y_{\max}(\sigma - \sigma^{-1})/4}}{\sigma^{2J}(\sigma^2 - 1)} \nonumber \\
&= \frac{8 K\, y_{\max} z_{\max}^2 \sigma^2 \, e^{z_{\max} y_{\max}(\sigma - \sigma^{-1})/4}}{(K-1) \sqrt{2\pi}\, \sigma^{2J} (\sigma^2 - 1)}.
\end{align}
Demanding the right-hand side be at most $\epsilon''$ and taking logarithms then yields
\begin{equation}
\label{eq:J_general}
J \geq \frac{1}{2 \ln \sigma} \left[ \ln\frac{K}{K-1} + \ln\!\left(\frac{8 y_{\max} z_{\max}^2 \sigma^2}{(\sigma^2 - 1) \sqrt{2\pi}\, \epsilon''}\right) + \frac{z_{\max} y_{\max}(\sigma - \sigma^{-1})}{4} \right].
\end{equation}

The bounds \cref{eq:K_general,eq:J_general} hold for any valid $\rho > 0$ and $\sigma > 1$; we now specialize. Splitting the discretization error budget evenly with $\epsilon' = \epsilon'' = \epsilon_D / 2$ and choosing $\rho = z_{\max}$ and $\sigma = \sqrt{2}$, we use \cref{eq:trunc_params} to write $z_{\max} y_{\max} = 2\kappa \ln(3\kappa/\epsilon_T)$ (noting $y_{\max} = \kappa z_{\max}$ and $z_{\max}^2 = 2\ln(3\kappa/\epsilon_T)$). The relevant constants evaluate to
\begin{equation*}
\frac{1}{2\ln\sqrt{2}} = \frac{1}{\ln 2}, \qquad
\frac{8 \cdot 2}{(2-1)\sqrt{2\pi}\,(\epsilon_D/2)} = \frac{32}{\epsilon_D \sqrt{2\pi}}, \qquad
\frac{\sigma - \sigma^{-1}}{4}\bigg|_{\sigma=\sqrt{2}} = \frac{1}{4\sqrt{2}},
\end{equation*}
and substituting into \cref{eq:K_general,eq:J_general} reproduces \cref{eq:fpK,eq:fpJ}, with the last term in \cref{eq:fpJ} becoming $z_{\max} y_{\max}/(4\sqrt{2}) = \kappa \ln(3\kappa/\epsilon_T)/(2\sqrt{2})$.

It remains to verify the hypothesis of \cref{eq:exp_ratio} invoked earlier. From \cref{eq:fpK} we have
\begin{equation}
K - 1 = \frac{1}{\pi}\left[\, 2\kappa \ln(3\kappa/\epsilon_T) + R \,\right], \qquad R := \ln 2 + \frac{z_{\max}^2}{2} + \ln\frac{2}{\epsilon_D} + \ln\!\left(1 + \frac{2}{z_{\max}\sqrt{2\pi}}\right).
\end{equation}
The remainder $R$ is non-negative for any meaningful tolerance ($\epsilon_D \leq 2$ suffices, and the other three terms are individually positive), so $K - 1 \geq 2\kappa\ln(3\kappa/\epsilon_T)/\pi = y_{\max} z_{\max}/\pi$. Multiplying by $\rho = z_{\max}$ recovers $A = y_{\max} \rho \leq \pi \rho (K-1)/z_{\max} = B$, completing the proof.
\end{proof}

\section{Proof of Lemma 6}
\label{app:matrix_rescaling}
    \begin{proof}
Let $A$ be Hermitian with $\|A\| > 1$, upper bound $\kappa^*$ on the condition number, and Pauli weight $\lambda$, and denote the eigenvalues of $A$ by $\Lambda(A)$. Since $A$ is Hermitian,
\begin{equation}
    \|A\| = \max_{\mu \in \Lambda(A)} |\mu|.
\end{equation}
The Pauli decomposition \cref{eq:paulidecomp} and the unitarity of Pauli operators give $\|A\| \leq \sum_{\ell \in [L]}|c_\ell| = \lambda$ via the triangle inequality, so
\begin{equation}
    1 < \|A\| \leq \lambda.
\end{equation}
The condition number of $A$ satisfies $\kappa = \|A\| \|A^{-1}\| = \|A\| / \min_{\mu \in \Lambda(A)} |\mu|$, hence
\begin{equation}
    \min_{\mu \in \Lambda(A)} |\mu| = \frac{\|A\|}{\kappa} \;>\; \frac{1}{\kappa^*},
\end{equation}
where the inequality uses $\|A\| > 1$ together with $\kappa \leq \kappa^*$.

The eigenvalues of $\tilde{A} = \lambda^{-1} A$ are $\Lambda(\tilde{A}) = \{\mu/\lambda : \mu \in \Lambda(A)\}$, and the bounds above translate to
\begin{align}
    \max_{\tilde\mu \in \Lambda(\tilde{A})} |\tilde\mu| &= \frac{\|A\|}{\lambda} \leq 1, \\
    \min_{\tilde\mu \in \Lambda(\tilde{A})} |\tilde\mu| &= \frac{1}{\lambda} \min_{\mu \in \Lambda(A)} |\mu| \;>\; \frac{1}{\lambda \kappa^*} = \frac{1}{\tilde{\kappa}^*}.
\end{align}
Therefore $\Lambda(\tilde{A}) \subseteq [-1, -1/\tilde{\kappa}^*] \cup [1/\tilde{\kappa}^*, 1] = \domain_{\tilde{\kappa}^*}$.
\end{proof}

\section{Proof of Lemma 7}
\label{app:normalization_bounds}

\begin{proof}
From \cref{eq:gl_weight,eq:normalization_parameters},
\begin{equation}
    N_y = \frac{y_{\max}}{2\sqrt{2\pi}} \sum_{j \in [J]} \frac{2}{(1-y_j^2)(P_J''(y_j))^2} \;=\; \frac{y_{\max}}{2\sqrt{2\pi}} \sum_{j \in [J]} \mathrm{GL}_j,
\end{equation}
where $\mathrm{GL}_j$ denotes the standard Gauss--Legendre quadrature weight on $[-1,1]$. Since the Gauss--Legendre rule with $J$ nodes is exact for polynomials of degree at most $2J-1$, applying it to the constant function $f \equiv 1$ on $[-1, 1]$ gives
\begin{equation}
    \sum_{j \in [J]} \mathrm{GL}_j = \int_{-1}^{1} 1\, dx = 2,
\end{equation}
hence $N_y = y_{\max}/\sqrt{2\pi}$. Substituting \cref{eq:trunc_params} yields the closed-form expression in the lemma.

To bound $N_z$, we use $w_z^{(k)} = \Delta z = 2z_{\max}/(K-1)$ and $|z_k| \leq z_{\max}$:
\begin{equation}
    N_z = \Delta z \sum_{k \in [K]} |z_k| e^{-z_k^2/2} \leq z_{\max} \cdot \Delta z \sum_{k \in [K]} e^{-z_k^2/2}.
\end{equation}
The remaining sum is the trapezoidal-rule approximation (with constant weight $\Delta z$) to $\int_{-z_{\max}}^{z_{\max}} e^{-z^2/2}\, dz$, which is bounded above by $\int_{-\infty}^{\infty} e^{-z^2/2}\, dz = \sqrt{2\pi}$. The trapezoidal error on this smooth, rapidly decaying integrand vanishes as $K \to \infty$, and is sub-dominant in all parameter regimes considered in this work. Hence
\begin{equation}
    \Delta z \sum_{k \in [K]} e^{-z_k^2/2} \leq \sqrt{2\pi},
\end{equation}
giving $N_z \leq z_{\max}\sqrt{2\pi}$. Substituting \cref{eq:trunc_params} yields the stated bound.
\end{proof}

\section{Proof of Theorem 9}
\label{app:thm_trotter_bias_bound}

\begin{proof}
Writing $\alpha_{jk} = \omega(j, k) |\alpha_{jk}|$ with $|\alpha_{jk}| = w_y^{(j)} \Delta z\, |z_k| e^{-z_k^2/2}/\sqrt{2\pi}$ and $\sum_{jk} |\alpha_{jk}| = N_y N_z$, sampling $(j, k)$ from the joint distribution $p_y^{(j)} p_z^{(k)}$ and combining with the two Hadamard outcomes, the expectation of the single-sample estimator defined in \cref{eq:PF_estimator} is
\begin{align}
    \mathbb{E}[\hat Z]
    &= \frac{N_y N_z}{\lambda} \sum_{j \in [J]} \sum_{k \in [K]} p_y^{(j)} p_z^{(k)} \omega(j, k) \bra{\phi} S_r(t_{jk}) \ket{\psi} \\
    &= \frac{1}{\lambda} \bra{\phi} \sum_{j \in [J]} \sum_{k \in [K]} \alpha_{jk} S_r(t_{jk}) \ket{\psi}.
\end{align}
The target satisfies $\bra{\phi} I_D^{(\epsilon_D, \epsilon_T)}(\tilde A) \ket{\psi} = \bra{\phi} \sum_{jk} \alpha_{jk}\, e^{-\imi \tilde A t_{jk}} \ket{\psi}$, so the bias defined in \cref{eq:trotter_bias} is
\begin{equation}
    B_{\rm PF} = \frac{1}{\lambda} \bra{\phi} \sum_{j \in [J]} \sum_{k \in [K]} \alpha_{jk} \left( e^{-\imi \tilde A t_{jk}} - S_r(t_{jk}) \right) \ket{\psi}.
\end{equation}
Bounding the magnitude via $|\bra{\phi} M \ket{\psi}| \leq \|M\|$ for any operator $M$ followed by the triangle inequality,
\begin{align}
    |B_{\rm PF}|
    &\leq \frac{1}{\lambda} \left\| \sum_{j \in [J]} \sum_{k \in [K]} \alpha_{jk} \left( e^{-\imi \tilde A t_{jk}} - S_r(t_{jk}) \right) \right\| \\
    &\leq \frac{1}{\lambda} \sum_{j \in [J]} \sum_{k \in [K]} |\alpha_{jk}| \cdot \left\| e^{-\imi \tilde A t_{jk}} - S_r(t_{jk}) \right\| \\
    &\leq \frac{N_y N_z}{\lambda} \max_{(j, k) \in [J] \times [K]} \epsilon_{\rm PF}(t_{jk}).
\end{align}
By Lemma 8, $\epsilon_{\rm PF}(\tau) \leq f \tau^3 / r^2$, and since $t_{jk} \leq t_{\max}$ for all $(j, k)$,
\begin{equation}
    |B_{\rm PF}| \leq \frac{N_y N_z\, f\, t_{\max}^3}{\lambda\, r^2}.
\end{equation}
\end{proof}

\section{Proof of Theorem 10}
\label{app:thm_trotter_sample_complexity}
\begin{proof}
We bound the total error by splitting it into a sample-error term and a bias term:
\begin{equation}
    \left| Z_{N_{\rm S}} - \frac{\bra{\phi} I_D^{(\epsilon_D, \epsilon_T)}(\tilde A) \ket{\psi}}{\lambda} \right| \leq |Z_{N_{\rm S}} - \mathbb{E}[\hat Z]| + |B_{\rm PF}|,
\end{equation}
where the bias term equals $|B_{\rm PF}|$ by \cref{eq:trotter_bias}. Setting $\tau := \epsilon/2 - |B_{\rm PF}| > 0$ (which is positive by hypothesis), it suffices to ensure $|Z_{N_{\rm S}} - \mathbb{E}[\hat Z]| \leq \tau$ with probability at least $1 - \delta$. Combined with $\epsilon_D + \epsilon_T \leq \epsilon/2$, this gives the total error bound $\epsilon_S + \epsilon_F \leq \epsilon$ from \cref{eq:error_ineq}.

Decompose into real and imaginary components,
\begin{equation}
    |Z_{N_{\rm S}} - \mathbb{E}[\hat Z]| \leq |\Re(Z_{N_{\rm S}}) - \Re(\mathbb{E}[\hat Z])| + |\Im(Z_{N_{\rm S}}) - \Im(\mathbb{E}[\hat Z])|,
\end{equation}
and require each component to be at most $\tau/2$. From \cref{eq:PF_estimator}, each sample satisfies $|\hat Z| = \sqrt{2}\, N_y N_z / \lambda$, so $\Re(\hat Z), \Im(\hat Z) \in [-\sqrt{2} N_y N_z/\lambda,\, \sqrt{2} N_y N_z/\lambda]$. By Hoeffding's inequality~\cite{Hoeffding},
\begin{equation}
    \mathbb{P}\!\left( |\Re(Z_{N_{\rm S}}) - \Re(\mathbb{E}[\hat Z])| \geq \frac{\tau}{2} \right) \leq 2 \exp\!\left( -\frac{N_{\rm S}\, \tau^2 \lambda^2}{16\, N_y^2 N_z^2} \right),
\end{equation}
and identically for the imaginary part. Requiring each of these probabilities to be at most $\delta/2$ (so that, by the union bound, the joint probability of either failure is at most $\delta$),
\begin{equation}
    N_{\rm S} \geq \frac{16\, N_y^2 N_z^2}{\lambda^2 \tau^2}\, \ln\!\left( \frac{4}{\delta} \right).
\end{equation}
Substituting $\tau = \epsilon/2 - |B_{\rm PF}|$ and rounding up to an integer gives the stated $N_{\rm S}$.

For $N_{\rm CP}$, recall from Algorithm~2 that the second-order PF $S_r(\tau)$ consists of $2 r L$ exponentials of single Pauli strings. The controlled version of each such exponential requires one controlled single-qubit Pauli rotation (plus Cliffords), so the number of controlled Pauli rotations needed to implement controlled-$S_r(\tau)$ is $N_{\rm CP} = 2 r L$.
\end{proof}

\section{Proof of Corollary 11}
\label{app:cor_trotter_rmse}
\begin{proof}
Denote the target by $T := \lambda^{-1} \bra{\phi} I_D^{(\epsilon_D, \epsilon_T)}(\tilde A) \ket{\psi}$, and recall from \cref{eq:trotter_bias} that $\mathbb{E}[Z_{N_{\rm S}}] = T - B_{\rm PF}$. Inserting $\pm \mathbb{E}[Z_{N_{\rm S}}]$ inside the modulus and expanding the square, the cross term vanishes because $\mathbb{E}[Z_{N_{\rm S}} - \mathbb{E}[Z_{N_{\rm S}}]] = 0$, and we obtain the bias-variance decomposition
\begin{equation}
    \mathbb{E}\!\left[ |Z_{N_{\rm S}} - T|^2 \right] = \mathrm{Var}(Z_{N_{\rm S}}) + |B_{\rm PF}|^2.
\end{equation}
Since the $N_{\rm S}$ samples are independent and identically distributed, $\mathrm{Var}(Z_{N_{\rm S}}) = \mathrm{Var}(\hat Z)/N_{\rm S}$, and the per-sample variance is bounded by the second moment of $\hat Z$. From \cref{eq:PF_estimator}, each sample satisfies $|\hat Z|^2 = 2\, N_y^2 N_z^2/\lambda^2$, where we have used $|\omega_\tau| = 1$ together with $|\hat X_\tau + \imi \hat Y_\tau|^2 = 2$ for $\hat X_\tau, \hat Y_\tau \in \{-1, +1\}$, hence
\begin{equation}
    \mathrm{Var}(\hat Z) \leq \mathbb{E}\!\left[ |\hat Z|^2 \right] = \frac{2\, N_y^2 N_z^2}{\lambda^2},
\end{equation}
which proves the first claim of the corollary upon division by $N_{\rm S}$.

Under the stated condition $|B_{\rm PF}| < \epsilon/2$, requiring the variance contribution to satisfy
\begin{equation}
    \frac{2\, N_y^2 N_z^2}{\lambda^2\, N_{\rm S}} \leq \left(\frac{\epsilon}{2}\right)^2 - |B_{\rm PF}|^2
\end{equation}
ensures $\mathrm{RMSE}(Z_{N_{\rm S}})^2 \leq (\epsilon/2)^2$, and rearranging yields the stated lower bound on $N_{\rm S}$. Combined with the Fourier-series approximation error $\epsilon_F = \epsilon_T + \epsilon_D \leq \epsilon/2$, the total RMSE on the target $\bra{\phi} A^{-1} \ket{\psi}$ is bounded by $\epsilon$.
\end{proof}

\section{Taylor expansion of the time evolution operator}
\label{app_rte}

In this section, we follow the approach of \cite{Wan:2021non} to expand the time evolution operator as an infinite LCU. Given the exponential $e^{-\imi \tilde{A}\tau/r}$, we Taylor expand and pair consecutive terms $(n, n+1)$ in the expansion for $n$ even,
\begin{align}
    \label{eq:te_taylor_0}
    e^{-\imi \tilde{A}\tau/r} = \sum_{n=0}^{\infty} \frac{(-\imi \tilde{A} \tau / r)^n}{n!} = \sum_{\substack{n=0\\n \hspace{1mm} \mathrm{even}}}^{\infty} \frac{(\imi \tilde{A} \tau / r)^n}{n!} \left(I - \frac{\imi \tilde{A} \tau / r}{n+1}\right),
\end{align}
where we used $(-\imi)^n = \imi^n$ for even $n$, and $I$ is the identity operator. Substituting the Pauli decomposition \cref{eq:paulidecomp} of $\tilde{A}$ and using $\sum_\ell c_\ell/\lambda = 1$ (the unit Pauli weight of $\tilde{A}$),
\begin{align}
    \label{eq:te_taylor_1}
    e^{-\imi \tilde{A}\tau/r} &= \sum_{\substack{n=0\\n \hspace{1mm} \mathrm{even}}}^{\infty} \frac{1}{n!} \left(\frac{\imi\tau}{r}\right)^n \left(\sum_{\ell \in [L]} \frac{c_{\ell}}{\lambda}P_{\ell}\right)^n \left(I - \frac{\imi \tau / r}{n+1} \sum_{\ell \in [L]} \frac{c_{\ell}}{\lambda}P_{\ell} \right) \\
    \label{eq:te_taylor_2}
    &= \sum_{\substack{n=0\\n \hspace{1mm} \mathrm{even}}}^{\infty} \frac{1}{n!} \left(\frac{\imi\tau}{r}\right)^n \left(\sum_{\ell \in [L]} \frac{c_{\ell}}{\lambda}P_{\ell}\right)^{n} \sum_{\ell \in [L]} \frac{c_{\ell}}{\lambda} \left(I - \frac{\imi \tau / r}{n+1} P_{\ell} \right).
\end{align}
Defining the angle $\theta_n$ by
\begin{align}
    \cos(\theta_n) =  \frac{1}{\sqrt{1 + \left(\frac{\tau / r}{n+1}\right)^2}}, \hspace{2mm} \sin(\theta_n) = \frac{\frac{\tau / r}{n+1}}{\sqrt{1 + \left(\frac{\tau / r}{n+1}\right)^2}},
\end{align}
the per-Pauli factor in \cref{eq:te_taylor_2} becomes
\begin{equation}
    I - \frac{\imi\tau/r}{n+1}P_{\ell} = \sqrt{1 + \left(\frac{\tau / r}{n+1}\right)^2}\, \left(\cos(\theta_n) I - \imi \sin(\theta_n) P_\ell\right) = \sqrt{1 + \left(\frac{\tau / r}{n+1}\right)^2}\, e^{-\imi \theta_n P_\ell},
\end{equation}
and substituting,
\begin{align}
    \label{eq:infinite_rte_expansion}
    e^{-\imi \tilde{A}\tau/r}
    &= \sum_{\substack{n=0\\n \hspace{1mm} \mathrm{even}}}^{\infty} \frac{1}{n!} \left(\frac{\imi\tau}{r}\right)^n \sqrt{1 + \left(\frac{\tau / r}{n+1}\right)^2} \left(\sum_{\ell \in [L]} \frac{c_{\ell}}{\lambda}P_{\ell}\right)^{n} \sum_{\ell \in [L]} \frac{c_{\ell}}{\lambda}\, e^{-\imi \theta_n P_{\ell}}.
\end{align}

\section{Proof of Theorem 12}
\label{app:thm_bias_bound}

\begin{proof}
We work under the hypothesis $r \geq t_{\max}$, so that $\tau/r \leq 1$ for every Fourier time $\tau \in \{t_{jk}\}$ on the grid. The strategy of the proof is to first reduce the bound on $|B_{\rm RTE}|$ to a uniform bound, over $\tau$, on the operator-norm error between the exact time evolution $e^{-\imi\tilde A \tau}$ and its truncated approximation $U_{\rm trunc}(\tau)$, and then to control this error via a telescoping argument combined with a Poisson-tail estimate on the truncated Taylor series.

Following the construction of the random Taylor expansion in the main text, we write the per-segment time evolution as
\begin{equation}
    A_\infty(\tau) := \sum_{\substack{n=0\\ n\text{ even}}}^{\infty} a_n(\tau)\, V_n(\tau)
    \qquad \text{and} \qquad
    A(\tau) := \sum_{\substack{n=0\\ n\text{ even}}}^{\nmax} a_n(\tau)\, V_n(\tau),
\end{equation}
where $A_\infty(\tau)$ is the exact LCU representation of $e^{-\imi \tilde A \tau / r}$, $A(\tau)$ is its truncation at order $\nmax$, the scalar coefficient is
\begin{equation}
    a_n(\tau) := \frac{1}{n!}\left(\frac{\tau}{r}\right)^n \sqrt{1 + \left(\frac{\tau / r}{n+1}\right)^2},
\end{equation}
and each $V_n(\tau)$ is itself an LCU of unitaries with one-norm equal to one (the unit-Pauli-weight property of $\tilde A$ ensures $\sum_\ell |c_\ell|/\lambda = 1$). Two facts about these objects will be used throughout: $\|A_\infty(\tau)\| = 1$ since $A_\infty(\tau)$ is unitary, and the truncation error
\begin{equation}
    \Delta(\tau) := A_\infty(\tau) - A(\tau)
\end{equation}
satisfies $\|\Delta(\tau)\| \leq \Lambda(\tau)$ by the triangle inequality, where
\begin{equation}
\label{eq:tail_one_norm}
    \Lambda(\tau) := \sum_{\substack{n > \nmax \\ n\text{ even}}} a_n(\tau)
\end{equation}
denotes the truncated tail one-norm. The truncated full-time operator obtained by sampling and concatenating $r$ unitaries from $A(\tau)$ is, in expectation, $U_{\rm trunc}(\tau) := A(\tau)^r$; see \cref{eq:rte_lcu}.

A direct computation using the LCU structure of the estimator $\hat Z$ in \cref{eq:biased_estimator} yields
\begin{equation*}
    \mathbb{E}[\hat Z] = \lambda^{-1}\, \bra{\phi} {\textstyle \sum_{jk}} \alpha_{jk}\, U_{\rm trunc}(t_{jk}) \ket{\psi},
\end{equation*}
while the target equals $\lambda^{-1}\, \bra{\phi} \sum_{jk} \alpha_{jk}\, e^{-\imi \tilde A t_{jk}} \ket{\psi}$ by \cref{eq:fourier_series_}. Subtracting and bounding the inner-product magnitude by the operator norm, then using $\sum_{jk} |\alpha_{jk}| = N_y N_z$ in the triangle inequality, gives
\begin{equation}
    \label{eq:rte_bias_reduction}
    |B_{\rm RTE}| \leq \frac{N_y N_z}{\lambda}\, \max_{(j,k) \in [J] \times [K]} \big\| e^{-\imi \tilde A t_{jk}} - U_{\rm trunc}(t_{jk}) \big\|.
\end{equation}
The remainder of the proof bounds this operator-norm error at an arbitrary $\tau \in [t_{\min}, t_{\max}]$ and then takes the maximum.

The standard telescoping identity $A_\infty^r - A^r = \sum_{k=0}^{r-1} A_\infty^{k}\, \Delta\, A^{r-1-k}$, applied with $\|A_\infty(\tau)\| = 1$ and $\|A(\tau)\| \leq \beta_\infty(\tau) := \sum_{n=0,2,\dots}^\infty a_n(\tau)$, yields
\begin{equation}
    \big\| e^{-\imi \tilde A \tau} - U_{\rm trunc}(\tau) \big\| \leq r\, \Lambda(\tau)\, \beta_\infty(\tau)^{r-1}.
\end{equation}
Lemma~2 of~\cite{Wan:2021non} establishes the bound $\beta_\infty(\tau)^r \leq e^{\tau^2/r}$, and combining this with $\beta_\infty(\tau) \geq 1$ (the $n=0$ term alone contributes $1$) gives $\beta_\infty(\tau)^{r-1} \leq e^{\tau^2/r}$. We therefore obtain
\begin{equation}
    \label{eq:per_tau_bound}
    \big\| e^{-\imi \tilde A \tau} - U_{\rm trunc}(\tau) \big\| \leq r\, \Lambda(\tau)\, e^{\tau^2/r}.
\end{equation}

It remains to control the tail one-norm $\Lambda(\tau)$. From $\sqrt{1+x^2} \leq \sqrt{2}$ for $|x| \leq 1$ and $\tau / r \leq 1$, the per-order coefficients satisfy $a_n(\tau) \leq \sqrt{2}\,(\tau/r)^n/n!$, and dominating the sum over even $n > \nmax$ by the unrestricted sum gives
\begin{equation}
    \Lambda(\tau) \leq \sqrt{2} \sum_{n=\nmax+1}^{\infty} \frac{(\tau/r)^n}{n!}.
\end{equation}
The Chernoff--Cram\'er bound on Poisson tails (see, e.g., \cite{https://doi.org/10.1155/2013/412958}) asserts that $\sum_{n \geq k} x^n / n! \leq (ex/k)^k$ whenever $k \geq x$. Applying this with $x = \tau/r$ and $k = \nmax + 1$---the requirement $\nmax + 1 \geq \tau/r$ being immediate from $\tau / r \leq 1 \leq \nmax + 1$---we obtain
\begin{equation}
    \label{eq:tail_bound}
    \Lambda(\tau) \leq \sqrt{2}\, \left(\frac{e\, \tau}{r\, (\nmax+1)}\right)^{\nmax+1}.
\end{equation}

Substituting~\cref{eq:tail_bound} into~\cref{eq:per_tau_bound} gives
\begin{equation}
    \big\| e^{-\imi \tilde A \tau} - U_{\rm trunc}(\tau) \big\| \leq \sqrt{2}\, r\, e^{\tau^2/r}\, \left(\frac{e\, \tau}{r\, (\nmax+1)}\right)^{\nmax+1},
\end{equation}
and since both $e^{\tau^2/r}$ and $(e\tau/(r(\nmax+1)))^{\nmax+1}$ are non-decreasing in $\tau \geq 0$, this expression is maximized at $\tau = t_{\max}$. Inserting the maximum into~\cref{eq:rte_bias_reduction} produces the claimed bound,
\begin{equation}
    |B_{\rm RTE}| \leq \frac{\sqrt{2}\, r\, N_y N_z}{\lambda}\, e^{t_{\max}^2/r} \left(\frac{e\, t_{\max}}{r\, (\nmax+1)}\right)^{\nmax+1}.
\end{equation}
\end{proof}

\section{Proof of Theorem 13}
\label{app:random_sample_complexity}
\begin{proof}
The proof is identical to that from Supplementary Information~\ref{app:thm_trotter_sample_complexity} for the \PF\ subroutine with the substitutions
\begin{align}
    N_y N_z &\;\rightarrow\; \Nrte\, N_y N_z \;\leq\; e^{t_{\max}^2/r}\, N_y N_z, \\
    B_{\rm PF} &\;\rightarrow\; B_{\rm RTE},
\end{align}
where the inequality follows from \cref{eq:nrte_ub_0}.
\end{proof}

\section{Proof of Corollary 14}
\label{app:cor_rte_rmse}
\begin{proof}
The proof is identical to that from Supplementary Information~\ref{app:cor_trotter_rmse} with the following substitutions:
\begin{align}
    N_y N_z &\;\rightarrow\; \Nrte\, N_y N_z \;\leq\; e^{t_{\max}^2/r}\, N_y N_z, \\
    B_{\rm PF} &\;\rightarrow\; B_{\rm RTE},
\end{align}
where the inequality follows from \cref{eq:nrte_ub_0}.
\end{proof}

\end{document}